\documentclass[runningheads]{llncs}

\input{preamble}

\newcommand\defaccr[2]{\newcommand#1{#2\xspace}}
\newcommand\defmath[2]{\newcommand#1{\ensuremath{#2}\xspace}}
\newcommand\concept[1]{\textit{#1}}

\renewcommand\dots{\makebox[.7em][c]{.\hfil.\hfil.}}

\defmath\PG{\mathcal{P}_n}

\usetikzlibrary{intersections,fit,shapes.misc, decorations.markings,patterns}

\makeatletter
\def\namedlabel#1#2{\begingroup
    #2%
    \def\@currentlabel{#2}%
    \phantomsection\label{#1}\endgroup
}
\makeatother

\makeatletter
\patchcmd{\ALG@step}{\addtocounter{ALG@line}{1}}{\refstepcounter{ALG@line}}{}{}
\newcommand{\ALG@lineautorefname}{Line}
\makeatother

\renewcommand\phi{\varphi}

\defmath\leqpoly{\leq_\poly}

\defmath{\img}{\mathtt{image}}
\defmath{\apre}{\mathtt{\forall preimage}}

\defmath{\Post}{\mathit{Postfix}}

\let\set\undefined

\providecommand{\set}[1]{\ensuremath{\left\lbrace #1 \right\rbrace}}

\providecommand{\sizeof}[1]{\ensuremath{\left\vert{#1}\right\vert}}

\providecommand{\gen}[1]{\ensuremath{\left\langle #1 \right\rangle}}

\defmath{\bool}{\ensuremath{\mathbb{B}}}
\defmath{\complex}{\ensuremath{\mathbb{C}}}
\defmath{\integers}{\ensuremath{\mathbb{Z}}}
\defmath{\conditionalind}{\mathrel{\text{\scalebox{1.07}{$\perp\mkern-10mu\perp$}}}}
\defmath{\dx}{\partial x}
\defmath{\ddx}{\sfrac{\partial}{\partial x}}
\defmath{\half}{\textstyle{\frac{1}{2}}}

\defmath\Exists{\mathit{Exists}}
\defmath\PlusExists{\mathit{PlusExists}}
\defmath\calciso{\mathsf{calciso}}

\newcommand{\defn}{\,\triangleq\,}

\tikzstyle{oval} = [state, ellipse, minimum size=4mm, inner sep=0.5mm, node distance=1cm]
\tikzset{every picture/.style={->,thick}}

\tikzstyle{leaf}=[draw, rectangle,minimum size=5mm, inner sep=3pt]
\tikzstyle{var}=[circle,draw=black!70,solid,thick,minimum size=6mm]
\tikzstyle{bdd}=[regular polygon, regular polygon sides=3, draw=black!70,solid,thick,inner sep=0.5mm]
\tikzstyle{n}=[->,loosely dashed,thick]
\tikzstyle{p}=[->,solid,thick]
\tikzstyle{b}=[->,densely dashdotted,ultra thick]

\defmath\before{\prec}
\defmath\beforeq{\preccurlyeq}

\newenvironment{smallmat}{\left[\begin{smallmatrix}}{\end{smallmatrix}\right]}

\newcommand{\Pbf}{\mathbf{P}}

\defaccr{\bdd}{\textsf{BDD}}
\defaccr{\bdds}{\textsf{BDD}s}
\defmath{\qmdd}{\textsf{SLDD}_{\times}}
\defaccr{\add}{\textsf{ADD}}
\defaccr{\isoqmdd}{\textsf{LIMDD}}
\defaccr{\limdd}{\textsf{LIMDD}}
\defaccr{\zlimdd}{\ensuremath{\braket{Z}}-\textsf{LIMDD}}
\defaccr{\paulilimdd}{Pauli-\limdd}
\defaccr{\paulilimdds}{Pauli-\limdds}
\defaccr{\qmdds}{\textsf{QMDD}s}
\defaccr{\adds}{\textsf{ADD}s}
\defaccr{\isoqmdds}{\textsf{LIMDD}s}
\defaccr{\limdds}{\textsf{LIMDD}s}
\defaccr{\glimdd}{\ensuremath{G}-\limdd}
\defaccr{\glimdds}{\ensuremath{G}-\limdds}

\defmath\oh{\mathcal O}

\defmath\rootlim{B_{\textnormal{root}}}
\defmath\lowlim{B_{\textnormal{low}}}
\defmath\highlim{B_{\textnormal{high}}}
\defmath\gmax{g}
\defmath\kmax{\kappa^{\textnormal{final}}}

\defmath\cast{\mathbb C^\ast}

\DeclareRobustCommand{\leafnode}[1][]{%
  \raisebox{-.8mm}{%
  \tikz{%
    \node[state,inner sep=0pt,minimum size=10pt,right= of x,leaf](v){\scriptsize $1$};%
  }%
  }%
}

\defmath\Low{\ensuremath{\textsf{low}}}
\defmath\High{\ensuremath{\textsf{high}}}
\defmath\LIM{\textsf{LIM}}

\tikzstyle{oval} = [state, ellipse, minimum size=4mm, inner sep=0.5mm, node distance=1cm]
\tikzset{every picture/.style={->,thick}}

\tikzstyle{leaf}=[draw, rectangle,minimum size=4.mm, inner sep=3pt]
\tikzstyle{var}=[circle,draw=black!70,solid,thick,minimum size=6mm]
\tikzstyle{bdd}=[regular polygon, regular polygon sides=3, draw=black!70,solid,thick,inner sep=0.5mm]
\tikzstyle{n}=[->,loosely dashed,thick]
\tikzstyle{p}=[->,solid,thick]
\tikzstyle{b}=[->,densely dashdotted,ultra thick]
\tikzset{every node/.style={initial text={}, inner sep=2pt, outer sep=0}}
\tikzstyle{e0}[0]=[dashed,thick,bend right=#1]
\tikzstyle{e1}[0]=[solid, bend left =#1]

\tikzstyle{lbl}=[draw,fill=white,inner sep=2pt, minimum size=0cm,line width=.5pt]

\defmath\yy{\begin{smallmat}
    0 & y^*\\
    y & 0\\
\end{smallmat}}

\defmath\ww{\begin{smallmat}
      0 & y   \\
      y^* & 0  \\
  \end{smallmat}
}

\tikzstyle{e0}[0]=[dotted,bend right=#1]
\tikzstyle{e1}[0]=[solid, bend left =#1]
\defmath\hv{{\hat v}}
\tikzset{every node/.style={initial text={}, inner sep=2pt, outer sep=0}}

\defmath\expsep{\succ\hspace{-1.5mm}\succ}

\defmath{\sumstate}{\ket{\text{Sum}}}
\defmath{\atleastassuccinctas}{\preceq_s}
\defmath{\strictlymoresuccinctthan}{\prec_s}
\defmath{\notmoresuccinctthan}{\npreceq_{s}}

\newcommand\hide[1]{}

\def\strictsuccinctto{
    \setbox0\hbox{
            $\longrightarrow$
    }\copy0\llap{\raise\ht0\hbox{
    {
    $    \hspace{0mm}\mathclap{\longleftarrow}{\hspace{-1.5mm}\times}\hspace{0mm}$
    }
    }}
}

\defmath\samp{\textbf{Sample}}
\defmath\pro{\textbf{Measure}}
\defmath\eq{\textbf{Equal}}
\defmath\res{\textbf{Res}}
\defmath\addi{\textbf{Addition}}
\defmath\inprod{\textbf{InnerProd}}
\defmath\fid{\textbf{Fidelity}}
\defmath\had{\textbf{Hadamard}}
\defmath\xyz{\textbf{X,Y,Z}}
\defmath\cx{\textbf{CX}}
\defmath\cz{\textbf{CZ}}
\defmath\swap{\textbf{Swap}}
\defmath\loc{\textbf{Local}}
\defmath\T{\textbf{T}}

\defmath\Rot{\mathit{Rot}}
\newcommand{\h}{\mathcal{H}}

\defmath\init{\mathit{init}}
\defmath{\qg}{\mathbf{G}}
\usepackage{enumerate}

\defmath\true{1}
\defmath\false{0}

\newcommand\no[1]{\ensuremath{\overline{#1}}}

\begin{document}

\title{Simulating~Quantum~Circuits~by~Model~Counting}

\authorrunning{J. Mei, M. Bonsangue \& A. Laarman}
\author{Jingyi Mei, Marcello Bonsangue and Alfons Laarman}
\institute{Leiden University}

\maketitle             
\begin{abstract}

    Quantum circuit compilation comprises many computationally hard reasoning tasks that 
    nonetheless lie inside \#$\P$ and its decision counterpart in~$\PP$.
    The classical simulation of general quantum circuits is a core example.
    We show for the first time that a strong simulation of universal quantum circuits can be efficiently tackled through weighted model counting by providing a linear encoding of Clifford+T circuits.
    To achieve this, we exploit the stabilizer formalism by Knill, Gottesmann, and Aaronson and the fact that stabilizer states form a basis for density operators.
    With an open-source simulator implementation, we demonstrate empirically that model counting often outperforms state-of-the-art simulation techniques based on the ZX calculus and decision diagrams.
    Our work paves the way to apply the existing array of powerful classical reasoning tools to realize efficient quantum circuit compilation; one of the obstacles on the road towards quantum supremacy.
\end{abstract}

\keywords{Quantum Computing \and Quantum circuit simulation \and Satisfiability \and \#SAT \and Weighted Model Counting \and Stabilizer Formalism.\vspace{-1em}}

\section{Introduction}

Classical simulation of quantum computing~\cite{viamontes2003improving,Anders2006fast,zulehner2018one,burgholzer2020improved} serves as a crucial task in the verification~\cite{ardeshir2014verification,thanos2023fast,peham2022equivalence,hong2022equivalence}, 
synthesis~\cite{brand2023quantum,schneider2023sat,amy2023symbolic} and optimization~\cite{Bravyi2021cliffordcircuit,shaik2023optimal} of quantum circuits.
In addition, improved classical simulation methods aid the search for a quantum advantage~\cite{huang2020classical,arute2019quantum}.

Due to the inherent exponential size of the underlying representations of quantum states and operations,
classical simulation of quantum circuits is a highly non-trivial task that comes in two flavors~\cite{estimation2020pashayan}:
Weak simulators only sample the probability distribution over measurement outcomes, i.e, they implement the ``bounded-error'' \BQP-complete problem that a quantum computer solves, whereas strong circuit simulators can compute the amplitude of any basis state, solving a \#\P-complete problem~\cite{van2010classical,jozsa2013classical}.

SAT-solvers~\cite{biere2009handbook,feng2023verification} have shown great potential to solve many formal methods problems efficiently in practice. %
Likewise, for the (weighted) model counting problem, i.e, counting number (or weights) of satisfying assignments,
many tools can manage large sets of constraints within industrial applications~\cite{oztok2015atopdown,sang2004combiningcc},
despite the \#P-Completeness of the problem.
Because strong quantum circuit simulation is \#\P-complete, weighted model counters are a natural fit that nonetheless has not yet been exploited.
To do so, we provide the first encoding of the strong simulation problem of universal quantum circuits
as a weighted model counting problem. Perhaps surprisingly, this encoding is linear in the number of gates or qubits, i.e., $\mathcal O(n+m)$ for $m$ gates and $n$ qubits.

One of the key properties of our encoding is that it does not use complex numbers to encode the probability amplitudes of quantum states ---which would prohibit the use of all modern model counters--- but merely requires negative weights.
We achieve this by exploiting a generalization~\cite{efficient2021zhang} of the Gottesman-Knill theorem~\cite{gottesman1997stabilizer}, which in effect rewrites the density matrix of any quantum state in the stabilizer state basis~\cite{gay2011stabilizer}, thereby obviating the need for complex numbers.
It turns out many exact weighted model counters do support negative weights out of the box. Our encoding thus empowers them to reason directly about  constructive interference, which is the cornerstone of all quantum algorithms.

Like many, our method builds on the Solovay–Kitaev theorem~\cite{dawson2005solovaykitaev} that in particular shows that the Clifford+$T$ gate set is universal for quantum computing, meaning that this gate set can efficiently approximate any unitary operator. Since our encoding also supports arbitrary rotation gates (phase shift $P$, $R_x$, $R_y$ and $R_z$), and since these rotations are non-Clifford gates, it can also support other universal gate sets like Clifford+$R_x$ or Clifford+$P$,
which allows for easier approximations.
One important example is Quantum Fourier Transformation (QFT),
which can be simulated with $O(n)$ rotation gates while needs at least $O(n\log(n))$ T gates to approximate~\cite{approximate2020nam}.
Moreover, since the hardness of exact reasoning about quantum circuits depends on the gate set (the classes \EQP{} and \NQP{} are  parameterized by the gate set as it defines the realizable unitaries), this flexibility greatly increases the strength of our strong simulation approach.

We implement our encoding in an open source tool QCMC, which uses the weighted model-counting tool GPMC~\cite{hashimoto2020gpmc}.
We demonstrate the scalability and feasibility of the proposed encoding through experimental evaluations based on three classes of benchmarks: random Clifford+$T$ circuits mimicking quantum chemistry applications~\cite{wright2022chemistry} and oracles, and quantum algorithms from MQT bench~\cite{mqt2023quetschlich}. 
We compare the results of our method against state-of-the-art circuit simulation tools QuiZX \cite{kissinger2022simulating} and Quasimodo~\cite{quasimodo}, respectively based on the ZX-calculus~\cite{coecke2011interactingZXAlgebra} and CFLOBDD~\cite{sistla2023weighted}. 
QCMC simulates important quantum algorithms like QAOA, W-state and VQE, etc,
which are not supported by QuiZX.
Additionally, QCMC outperforms Quasimodo on almost all random circuits and uses orders of magnitudes less memory than Quasimodo on all benchmarks.

In sum, this paper makes the following contributions.

\begin{itemize}[noitemsep,topsep=0pt,parsep=0pt,partopsep=0pt]
    \item a generalized stabilizer formalism formulated in terms of stabilizer groups, which forms a basis for our encoding;
    \item  the first encoding for circuits in various universal quantum gate sets as a weighted model counting problem, which is also linear in size;
    \item an implementation based on the weighted model counting tool GPMC;
    \item new benchmarks for the WMC competition~\cite{modelcounting2023}, and insights on improving model counters and samplers for applications in quantum computing.
\end{itemize}

\section{Preliminaries}
\label{sec2}
\label{sec:preliminaries}

\subsection{Quantum Computing}\label{sec21}

Similar to bits in classical computing, 
a quantum computer operates on quantum bits, or short as \emph{qubits}.
A bit is either $0$ or $1$, while
a qubit has states $\ket{0}$ or $\ket{1}$.
Here `$\ket{ }$' is the \emph{Dirac notation}, standard in quantum computing,
standing for a unit column vector, i.e.,
 $\ket{0} = [1,0]^\mathrm{T}$ and $\ket{1} = [0,1]^\mathrm{T}$, while $\bra{\psi}$ denotes the complex conjugate and transpose of $\ket{\psi}$, that is a row vector: $\bra{\psi} = \ket{\psi}^\dag$.

Let $\h$ be a Hilbert space.
A $n$-qubit quantum state is a $2^n$-dimension unit column vector in $\h$.
In the case of $n = 1$, a pure state $\ket{\psi}$ is  written as $\ket{\psi} = \alpha\ket{0} + \beta\ket{1}$, where $\alpha$ and $\beta$ are complex numbers in $\complex$ satisfying $|\alpha|^2 + |\beta|^2 = 1$.
Sometimes we represent a pure quantum state $\ket{\psi}$ by the density operator obtained as the product $\op{\psi}{\psi}$ of the state with itself.
For the rest of the paper,
we fix $n$ to be the number of the qubits.

Operations on quantum states are given by \emph{quantum gates}.
For an $n$-qubit quantum system,
a (global) quantum gate is a function $\mathbf{G}: \h\rightarrow\h$,
which can be described by $2^n\times 2^n$ \emph{unitary matrix} $U$,
i.e., with the property that $UU^\dag = U^\dag U = I$.
A quantum gate is \emph{local} when it works on a subspace of quantum system,
which can be extended to a \concept{global quantum gate} by apply identity operators on unchanged qubits,
i.e. a local quantum gate $U$ on qubit $i$ can be represented as a global quantum gate $U_i = \underbrace{I\otimes\cdots \otimes I}_{i \textup{ terms}} \otimes\,\ U \otimes \underbrace{I \otimes \cdots\otimes I}_{n-i-1\textup{ terms}}$.
Examples of quantum gates are the following $2 \times 2$ \emph{Pauli matrices} (or \emph{Pauli gates}):
\[
  \begin{aligned}
    \sigma[00] \equiv \,I\, \equiv
    \begin{bmatrix}
      1 & 0 \\
      0 & 1
    \end{bmatrix},
    \sigma[01] \equiv Z \equiv
    \begin{bmatrix*}[r]
      1 & 0 \\
      0 & -1
    \end{bmatrix*},
    \sigma[10] \equiv X \equiv
    \begin{bmatrix*}[r]
      0 & 1 \\
      1 & 0
    \end{bmatrix*},
    \sigma[11] \equiv Y \equiv
    \begin{bmatrix*}[r]
      0 & -\dot{\imath} \\
      \dot{\imath} & 0
    \end{bmatrix*}.
  \end{aligned}
\]
Pauli matrices form a basis for  $2 \times 2$ Hermitian matrix space. 
Let $\mathrm{PAULI}_n$ be the set of the tensor product of $n$ Pauli operators (a ``Pauli string'').
Any density operator in an $n$-qubit system can be written as $\sum_i \alpha_i P_i$ where $P_i\in \mathrm{PAULI}_n$ and $\alpha_i$ are real numbers~\cite{gay2011stabilizer}.
The so-called Pauli group is generated by multiplication of the local Pauli operators $X,Z$ as follows: $\PG = \gen{X_0, Z_0 \dots, X_{n-1}, Z_{n-1}}$.
Structurally, the Pauli group can now be written as $\PG = \set{k \cdot P  \mid P\in \mathrm{PAULI}_n , k \in \set{\pm 1, \pm\dot{\imath}} }$.
For instance, we have $-X\otimes Y \otimes Z \in \mathcal{P}_3$.

The evolution of a quantum system can be modeled by a quantum circuit.
Let $[m]$ be $\{0,\dots, m-1\}$.
A quantum circuit is given by a sequence of quantum gates: $C\equiv \mathbf{G}^0 \cdots\mathbf{G}^{m-1}$, where $\mathbf{G}^t \in \mathcal H$ is a global quantum gate at time step $t\in [m]$.
Let $U^t$ be the unitary matrix for $\qg^t$.
Then $C$ is represented by the unitary matrix $U = U^{m-1}\cdots U^0$.

An important class of quantum circuits is the so-called Clifford class or group,
as they can describe interesting quantum mechanical phenomena such as entanglement, teleportation, and superdense encoding.
More importantly, they are widely used in quantum error-correcting codes~\cite{calderbank1996quantum,steane1996error}
and measurement-based quantum computation~\cite{raussendorf2001oneway}.
The Clifford group is the set of unitary operators that map the Pauli group to itself through conjugation,
i.e. all the $2^n \times 2^n$ unitary matrices $U$ such that $UPU^\dag\in \PG$ for all $P\in\PG$.
It is generated by the local Hadamard ($H$) and phase ($S$) gates, 
and the two-qubit control-not gate~($CX$,~$CNOT$):
\[
  \begin{aligned}
    H = \frac{1}{\sqrt{2}}
    \begin{bmatrix*}[r]
      1 & 1 \\
      1 & -1
    \end{bmatrix*},
  \quad
  S =
  \begin{bmatrix}
    1 & 0 \\
    0 & \dot{\imath}\
  \end{bmatrix},
  \quad
  \text{and }
  CX = 
  \begin{smallmat}
    1 & 0 & 0 & 0 \\
    0 & 1 & 0 & 0 \\
    0 & 0 & 0 & 1 \\
    0 & 0 & 1 & 0 \\    
  \end{smallmat}.
  \end{aligned}
\]
Recall that $U_j$ performs $U$ on $j$-th qubit.
Similarly, we denote 
 by $CX_{ij}$ the unitary operator taking the $i$-th qubit as the control qubit and $j$-qubit as the target to execute a controlled-not gate.
Clifford circuits are circuits only containing gates from the Clifford group.

A (projective) measurement is given by a set of \emph{projectors} $\{\mathbb{P}_0,\dots,\mathbb{P}_{k-1}\}$ ---one for each measurement outcome $[k]$--- satisfying $\sum_{j\in[k]}\mathbb{P}_j = I$.
A linear operator $\mathbb{P}$ is a projector if and only if $\mathbb{P}\mathbb{P}=\mathbb{P}$.
For example,
given a 3-qubit system,
measuring the first 2 qubits under computational basis is given by projective measurement
$\set{~~
   \op{0}{0}\otimes\op{0}{0}\otimes I,~~
   \op{1}{1}\otimes\op{0}{0}\otimes I,~~
   \op{0}{0}\otimes\op{1}{1}\otimes I,~~
   \op{1}{1}\otimes\op{1}{1}\otimes I
   ~~}$.

Weak simulation is the problem of sampling the measurements outcomes according to the probability distribution induced by the semantics of the circuit.
In this work, we focus on strong simulation as defined in \autoref{def:simulation}.
Here we assume, without loss of generality, that a circuit is initialized to the all-zero state: $\ket{0}^{\otimes n}$.

\begin{definition}[The strong simulation simulation]\label{def:simulation}
Given an $n$-qubit quantum circuit $C$ and a measurement  $M=\{\mathbb{P}_0,\dots,\mathbb{P}_{k-1}\}$,
a strong simulation of circuit $C$ computes the probability of getting any outcome $l\in [k]$,
that is, the value $\bra{0}^{\otimes n}C^\dag\mathbb{P}_l C\ket{0}^{\otimes n}$, up to a  number of  desired bits of precision.
\end{definition}

The Gottesman-Knill theorem~\cite{gottesman1997stabilizer} shows that  Clifford circuits can be strongly stimulated by classical algorithms in polynomial time and space.

\subsection{Stabilizer Groups}\label{sec22}

The stabilizer formalism~\cite{gottesman1997stabilizer} is a subset of quantum computing that can be effectively simulated on a classical computer.
A state $\ket{\varphi}$ is said to be \emph{stabilized} by a quantum unitary operator $U$ if and only if it is a +1 eigenvector of $U$, i.e., $U\ket{\varphi} = \ket{\varphi}$. 
For example, we say $\ket 0$ is \concept{stabilized} by $Z$ as $Z\ket 0 = \ket 0$. Similarly, $\ket{+}=  \frac{1}{\sqrt{2}}\ket{0} + \frac{1}{\sqrt{2}}\ket{1}$ is stabilized by $X$, and all states are stabilized by $I$. 
The \emph{stabilizer states} form a strict subset of all quantum states which can be uniquely described by maximal commutative subgroups of the Pauli group $\PG$, which is called \emph{stabilizer group}.
The elements of the stabilizer group are called \emph{stabilizers}.
Recall that the Clifford group is formed by unitary operators mapping the Pauli group to itself.
This leads to the fact that stabilizer states are closed under operators from the Clifford group.

Given an $n$-qubit stabilizer state $\ket{\psi}$,
let $\mathcal{S}_{\ket{\psi}}$ be the stabilizer group of $\ket{\psi}$.
While the elements of a Pauli group $\PG$ either commute or anticommute, 
a stabilizer group $\mathcal{S}$ must be abelian, 
because if $P_1, P_2\in\mathcal{S}_{\ket{\phi}}$ anticommute, i.e, $P_1P_2 = -P_2P_1$,
there would be a contradiction: $\ket{\phi} = P_1P_2\ket{\phi} = -P_2P_1\ket{\phi} = - \ket{\phi}$.
In particular, $-I^{\otimes n}$ can never be a stabilizer.
In fact,
a subgroup $\mathcal{S}$ of $\PG$ is a stabilizer group for an $n$-qubit quantum state if and only if it is an abelian group without $-I^{\otimes n}$.
Therefore, for any  Pauli string $\mathbf{P}\in \mathrm{PAULI}_n$, if $k\mathbf{P}\in\mathcal{S}$, we have $k=\pm 1$, since
$\forall k\mathbf{P}\in\mathcal{S}: \quad (k\mathbf{P})\ket{\phi} = (k\mathbf{P})^2\ket{\phi} = k^2 I \ket{\phi} = k^2\ket{\phi} \Rightarrow k = \pm 1.$
For clarity, we will use the symbols $P$ for Pauli strings with weight and use symbol $\mathbf{P}$ for Pauli strings without weight, i.e. $\mathbf{P}\in \mathrm{PAULI}$.

As every subgroup of a free group is free, 
any stabilizer group $\mathcal{S}$ can be specified by a set of generators $\mathcal{G}$ so that every element in $\mathcal{S}$ can be obtained through matrix multiplication on $\mathcal{G}$,
denoted as $\gen{\mathcal{G}}= \mathcal{S}$. 
The set of generators $\mathcal{G}$ needs not to be unique and has order $|G| = n$,
where $n$ represents the number of qubits,
and the corresponding stabilizer group $\mathcal{S}$ has order $2^{|G|}$. 

\begin{example}\label{ex21}
  The Bell state $\ket{\Phi_{00}} = \frac{1}{\sqrt{2}}(\ket{00} + \ket{11})$ can be represented by the following stabilizer generators written in square form:
    \[
   \pushQED{\qed}
      \frac1{\sqrt 2}(\ket{00} + \ket{11}) \equiv 
                          \gen{\begin{matrix}
                          X\otimes X \\
                          Z\otimes Z \\
                          \end{matrix}} 
                          \equiv 
                          \gen{\begin{matrix*}[r]
                          X\otimes X \\
                         -Y\otimes Y \\
                          \end{matrix*}}.
   \qedhere
   \popQED
  \]
\end{example}

We can relate the (generators of the) stabilizer group directly to the stabilizer state $\ket{\psi}$, as
the density operator of the stabilizer state can be written in a basis of Pauli matrices as follows~\cite{entanglement2005toth}:
\begin{equation} \label{eq:densitymatrix}
    \op{\psi}{\psi} = \prod_{G\in\mathcal{G}_{\ket{\psi}}} \frac{I+G}{2} = 
    \frac{1}{2^n}\sum_{P\in\mathcal{S}_{\ket{\psi}}} P.
\end{equation}
If a Clifford gate $U$ is applied to the above state, i.e. $U\ket{\psi}$,
and let $P\in \mathcal{S}_{\ket{\psi}}$,
the corresponding stabilizers of $U\ket{\psi}$ can be obtained by $UPU^\dag$ since $UPU^\dag U\ket{\psi} = U\ket{\psi}$.
Thus we have $\mathcal{S}_{U\ket{\psi}} = \{UPU^\dag \ | \ P\in \mathcal{S}_{\ket{\psi}}\}$ and the density operator of the resulting state will be obtained by conjugating $U$ on $\op{\psi}{\psi}$,
i.e. $U\op{\psi}{\psi}U^\dag = \frac{1}{2^n}\sum_{P\in\mathcal{S}_{\ket{\psi}}} UPU^\dag$.
To be specific,
consider performing a Clifford gate $U_j$, denoting a single qubit gate $U$ applied to $j$-th qubit as given in previous section,
to a stabilizer $P = P_1\otimes \dots \otimes P_n$ where $P \in \mathcal{S}_{\ket{\psi}}$.
We have $U_jPU_j^\dag = P_1\otimes \dots \otimes UP_jU^\dag\otimes \dots \otimes P_n$.
Since $U$ is a Clifford gate, $UP_jU^\dag \in \PG$.  %
Thus only the sign $k$ and the $j$-th entry need to be updated,
which can be done in constant time following the rules in \autoref{tab:clifford} for $H$ and $S$.
Applying a two-qubit gate $CX_{ij}$ is similar to updating the sign and Pauli matrices in $i$-th and $j$-th position (see \autoref{tab:clifford}),
which also takes constant time.
Overall, updating all generators after performing one Clifford gate can be done in $O(n)$ time.

\begin{table}[t!]
  \centering
  \caption{Lookup table for the action of conjugating Pauli gates by T gates. The subscripts ``c'' and ``t'' stand for ``control" and ``target". Adapted from~\cite{garcia2014geometry}.}
  \label{tab:clifford}
  \setlength{\tabcolsep}{5pt} 
  \begin{tabularx}{0.65\textwidth}{c|rr||c|cc}
      \toprule
      \textbf{Gate} & \textbf{In} & \textbf{Out} & \textbf{Gate} & \textbf{In} & \textbf{Out} \\
      \midrule
      & $X$ & $Z$ & \multirow{6}{*}{CX} & $\phantom{-}I_c \otimes X_t$ & $\phantom{-}I_c\otimes X_t$ \\
      $H$ & $Y$ & $-Y$ & & $\phantom{-}X_c \otimes I_t$ & $\phantom{-}X_c \otimes X_t$ \\
      & $Z$ & $X$ & & $\phantom{-}I_c \otimes Y_t$ & $\phantom{-}Z_c \otimes Y_t$ \\
      \cline{1-3}
       & $X$ & $Y$ & & $\phantom{-}Y_c \otimes I_t$ & $\phantom{-}Y_c \otimes X_t$ \\
      $S$ & $Y$ & $-X$ & & $\phantom{-}I_c \otimes Z_t$ & $\phantom{-}Z_c \otimes Z_t$ \\
      & $Z$ & $Z$ & & $\phantom{-}Z_c \otimes I_t$ & $\phantom{-}Z_c \otimes I_t$ \\
      \bottomrule
  \end{tabularx}
  \vspace*{-1em} 
\end{table}

Although stabilizer states are a subset of quantum states,
the space of density matrices for $n$-qubit states has a basis consisting of density matrices of $n$-qubit stabilizer states~\cite{gay2011stabilizer},
which makes it possible to extend the stabilizer formalism to a general quantum state.
\begin{example}\label{ex20}
  The Bell state $\ket{\Phi_{00}} = \frac{1}{\sqrt{2}}(\ket{00} + \ket{11})$ is a stabilizer state, 
  as it can be obtained by the following circuit, which evaluates to $CX_{01} \cdot H_0 \cdot \ket{00} = \ket{\Phi_{00}}$:
  \[
    \begin{array}{c}  
      \Qcircuit @C=1em @R=.7em {
        \lstick{\ket{0}} & \qw\ar@{.}[]+<0em,1em>;[d]+<0em,-0.5em> & \gate{H} &\qw\ar@{.}[]+<0em,1em>;[d]+<0em,-0.5em> & \ctrl{1} & \qw\ar@{.}[]+<0em,1em>;[d]+<0em,-0.5em> &\qw \\
        \lstick{\ket{0}} & \qw & \qw & \qw & \targ & \qw &\qw\\
        & \ket{\varphi_0} & & \ket{\varphi_1} & & \ket{\varphi_2} &
      }
  \end{array}
  \]
  We can simulate the above circuit with the stabilizer formalism.
  The stabilizer generator set for each time step can be obtained using the rules shown in \autoref{tab:clifford}:
  \[
  \begin{aligned}
      \gen{\begin{matrix}
    Z\otimes I \\
    I \otimes Z
    \end{matrix}} 
    &\xrightarrow{  H_0} 
    \gen{\begin{matrix}
      HZH^\dag \otimes I \\
      HIH^\dag \otimes Z
      \end{matrix}} 
      =
    \gen{\begin{matrix}
        X \otimes I \\
        I \otimes Z
      \end{matrix}} 
      \xrightarrow{ CX_{0,1}} 
      \gen{\begin{matrix}
        CX(X \otimes I)CX^\dag \\
      CX (I \otimes Z)CX^\dag
        \end{matrix}}
        =
        \gen{\begin{matrix}
      X \otimes X \\
      Z \otimes Z
      \end{matrix}}.
  \end{aligned}
    \]

    \vspace{-2.3em} \qed
\end{example}

According to the Gottesmann-Knill Theorem,
a quantum circuit can be effectively simulated when it starts from a stabilizer state and contains exclusively Clifford gates~\cite{gottesman1997stabilizer}. 
Because of the small number of generators needed, 
stabilizer groups can be used for an effective classical simulation of quantum states. 
In a $2^n$-dimensional Hilbert space,
instead of describing a quantum state by a $2^n$ complex-valued vector, one can use its stabilizer generator set to indicate the quantum state.

It's important to note that Clifford gates don't constitute a universal set of quantum gates. 
Certain unitary operators, 
like the $T$ gate represented by $T = \op{0}{0} + e^{\dot{\imath}\pi/4}\op{1}{1}$, 
fall outside this set.
However, by augmenting the set of Clifford gates with the $T$ gate, 
it becomes possible to approximate any unitary operator with arbitrary accuracy, 
as shown in~\cite{Kliuchnikov2013synthesis,bravyi2016improved}.

\subsection{Weighted Model Counting}

Weighted model counters can solve probabilistic reasoning problems like Bayesian inference~\cite{chavira2008probabilistic}.
In these ``classical'' settings, the weights in the encoding of e.g. a Bayesian network represent positive probabilities. 
The quantum setting that we study here is well-known for its complex probabilities, called amplitudes.
However, by expressing (pure) quantum states as their density matrices, and rewriting those in the stabilizer basis~\cite{gay2011stabilizer}, we obviate the need for complex numbers, as shown in the next section.
It turns out that existing weighted model counting tools, like GPMC~\cite{hashimoto2020gpmc}, already support negative weights %
(see \autoref{S5}).

Let $\bool$ be the Boolean set $\set{0,1}$. A propositional formula $C \colon \bool^V \to \bool$ over a finite set of Boolean variables $V$ is \emph{satisfiable} if there is an assignment $\alpha \in \bool^V$ such that $C(\alpha) = \true$. We define the set of all satisfiable assignments of a propositional formula $C$ as $SAT(C):= \set{\alpha \mid \alpha(C) = 1}$.
We write an assignment $\alpha$ as a \concept{cube} (a conjunction of literals, i.e. positive or negative variables), e.g., $a\land b$, or shorter~$ab$.

A weight function $W\colon \set{ \no v, v \mid v\in V} \to \mathbb{R} $ assigns to the positive literal $v$ (i.e., $v = \true$) and the negative literal $\no v$
(i.e., $v = \false$)
a real-valued weight.
We say variable $v$ is \emph{unbiased} iff $W(v) = W(\no v) = 1$.
Given an assignment $\alpha\in\bool^V$,
let $W(\alpha(v)) = W(v = \alpha(v))$ for $v\in V$.

For a propositional formula $C$ over variables in $V$ and weight function $W$, 
we define \concept{weighted model counting} as follows.
\[
MC_W(C) \defn   \sum_{\alpha \in  \bool^V} C(\alpha)\cdot W(\alpha)\text{, where }  W(\alpha) = \prod_{v\in V} W(\alpha(v)).
\]

\begin{example}
  Given a propositional formula $C = v_1 v_2 \vee \no{v_1} v_2 \vee v_3$ over $V=\{v_1,v_2, v_3\}$, there are two satisfying assignments: $\alpha_1 = v_1  v_2  v_3$ and $\alpha_2 = \no{v_1}  v_2  v_3$. 
  We define the weight function $W$ as $W(v_1) = -\frac{1}{2}$, $W(\no{v_1}) = \frac{1}{3}$ and $W(v_2) = \frac14$, $W(\no{v_2}) = \frac34$, while $v_3$ remains unbiased.
  The weight of $C$ can be computed as $MC_W(C) = -\frac{1}{2} \times \frac14  \times 1 + \frac{1}{3}  \times \frac14  \times 1  = -\frac{1}{24}$.
  \qed
\end{example}

\section{Encoding Quantum Circuits as Weighted CNF}
\label{sec3}
\label{sec:encoding}

\subsection{Generalized stabilizer formalism}

Clifford circuits together with $T$ gates generate states beyond stabilizer states,
enabling universal quantum computation.
As is shown in \autoref{tab:clifford},
Clifford gates map the set of Pauli matrices to itself,
keeping stabilizers within the Pauli group.
In contrast, 
$T$ gates can transform a Pauli matrix into a linear combination of Pauli matrices.
To be specific,
\autoref{tab:cliffordt} gives the action of $T$ gates on different Pauli gates.
\begin{table}[b!]
  \caption{Lookup table for the action of conjugating Pauli gates by $T$ gates. }
  \label{tab:cliffordt}
  \setlength{\tabcolsep}{15pt} 
  \centering
      \begin{tabularx}{0.63\textwidth}{c| c c c}
    \toprule
    \textbf{Gate} & \textbf{In} & \textbf{Out} & \textbf{Short} \\
    \midrule
      & $X$ &  $\frac1{\sqrt{2}} ({X + Y})$  & $\defn X'$ \\
    $T$ & $Y$ &  $\frac1{\sqrt{2}} ({Y - X})$& $\defn Y'$ \\
    & $Z$ & $Z$  &\\
    \bottomrule
  \end{tabularx}
\end{table}
Given a Pauli string, after performing Clifford+$T$ gates,
we will get a summation of weighted Pauli strings, e.g,
$T_1 \cdot (X\otimes X) \cdot T_1^\dag = \frac{1}{\sqrt{2}}(X\otimes X+Y \otimes X)$.
This leads to the definition of \emph{generalized stabilizer state} extended from standard stabilizer formalism.
\begin{definition}
    In a $n$-qubit quantum system,
    a generalized stabilizer state $\ket{\psi}$ is the simultaneous eigenvector, with eigenvalue 1, of a group containing $2^n$ commuting unitary operators $S$.
    The set of $S$ is a generalized stabilizer group.
\end{definition}

The above definition is adapted from \cite{efficient2021zhang} by defining a generalized stabilizer state using a generalized stabilizer group instead of generalized stabilizer generators.
In this way, we can easily get the corresponding stabilizer state by a weighted summation of its stabilizers to avoid multiplications between Pauli strings (the middle part of \autoref{eq:densitymatrix}).
The following proposition is also adapted from \cite{efficient2021zhang}, where they demonstrate that any pure state can be uniquely described by 
a set of stabilizers.
We additionally show that there exists a set of stabilizers, forming a group, which uniquely describes any pure state.

\begin{proposition}~\label{prop:pauli-basis}
    For any pure state $\ket{\varphi}$ in a $n$-qubit quantum system,
    there exists a generalized stabilizer group $\mathcal{S}_{\ket{\varphi}}$ such that
    $\op{\varphi}{\varphi} = \frac1{2^n} \cdot \sum_{P\in\mathcal{S}_{\ket{\varphi}}} P$.
\end{proposition}
\begin{proof}
    Given a pure state $\ket{\varphi}$, we have $\ket{\varphi} = U\ket{0}^{\otimes n}$,
    where $U$ is a unitary operator.
    Let $\mathcal{S}_{\ket{\varphi}} = \{ USU^\dag \mid S \in \mathcal{S}_{\ket{0}^{\otimes 0}} \}$,
    which is an isomorphic group to $\mathcal{S}_{\ket{0}}$ since $U$ is unitary.
    For any $S\in \mathcal{S}_{\ket{\varphi}}$, 
    we have $S'\in \mathcal{S}_{\ket{0}^{\otimes n}}$ satisfying $S = US'U^\dag$ and $S\ket{\varphi} = US'U^\dag U\ket{0}^{\otimes n} = US'\ket{0}^{\otimes n} = U\ket{0}^{\otimes n} = \ket{\varphi}$.
    Hence $S_{\ket{\varphi}}$ is the generalized stabilizer group of state $\ket{\varphi}$.
    Furthermore,
    we have $\op{\varphi}{\varphi} = U\op{0}{0}U^\dag = U(\frac{1}{2^n}\sum_{S'\in\mathcal{S}_{\ket{0}^{\otimes n}}} S')U^\dag = \frac{1}{2^n}\sum_{S'\in\mathcal{S}_{\ket{0}^{\otimes n}}} U S' U^\dag = \frac{1}{2^n} \cdot \sum_{S\in\mathcal{S}_{\ket{\varphi}}} S$.
    In fact, since any generalized stabilizer $S \in \mathcal{S}_{\ket{\varphi}}$ is a Hermitian matrix, as $S = US'U^\dag = US'^\dag U^\dag = S^\dag$ with $S'\in\mathcal{S}_{\ket{0}}$,
    it can be written as $S = \sum_i \alpha_i P_i$ where $P_i \in \textrm{PAULI}_n$ and $\alpha_i$ are real numbers satisfying $\sum_i \alpha_i^2 = 1$.
    Thus a generalized stabilizer is always a linear combination of stabilizers.
  \end{proof}

\defmath\GStab{\mathit{GStab}}
\defmath\stab{\mathit{Stab}}

For a $n$-qubit quantum space,
let $\GStab$ be the set of generalized stabilizer states, which is also the set of all pure states.
Let $\mathcal{Q}$ be the set of quantum states generated from a Clifford+$T$ circuit starting from the all-zero state,
i.e. $\mathcal{Q} = \{U_{m-1} \dots U_0\ket{0}^{\otimes n} \mid U_i\in \{H,S,CX,T\}\}$.
We have $\stab \subset \mathcal{Q} \subset \GStab$ and any pure state in $\GStab$ can be approximated by some state in $\mathcal{Q}$ with arbitrary accuracy~\cite{dawson2005solovaykitaev}.
For any $\ket{\varphi}\in \stab$,
we have $\mathcal{S}_{\ket{\varphi}} = \{P \mid P\in \mathrm{PAULI}_n\}$.
For any $\ket \psi \in \GStab$, we have $\mathcal{S}_{\ket{\psi}} = \set{ \sum_i \alpha_i P_i \mid P_i\in \mathrm{PAULI}_n,  \alpha_i \in \mathbb{R}}$.
E.g., when applying Clifford+$T$ gates on $\ket{\varphi}$,
we get generalized stabilizer state~$\ket{\psi}$ with $\mathcal{S}_{\ket{\psi}} = \set{ \sum_i \pm(\frac{1}{\sqrt{2}})^{k_i} P_i\ | \ P_i\in \mathrm{PAULI}_n,  k_i \in \mathbb{N}^+}$ (by updating $\mathcal{S}_{\ket{\varphi}}$ based on \autoref{tab:clifford} and \autoref{tab:cliffordt}).
Combining with \autoref{prop:pauli-basis}, we may flatten summations:
\begin{equation} \label{eq:gendensity}
    \op{\psi}{\psi} = 
     \frac{1}{2^n}\sum_{S\in\mathcal{S}_{\ket{\psi}}} \sum_i \pm(\frac{1}{\sqrt{2}})^{k_i} P_i = \sum \pm(\frac{1}{\sqrt{2}})^{k_i} P_i.
\end{equation}
Hence the density operator is determined by a summation of weighted Pauli strings and there is no need to distinguish which of the $2^n$ generalized stabilizers a Pauli string belongs to. We exploit this in our encoding and \autoref{ex35}.

\begin{example}\label{ex35}
 Reconsider the circuit in \autoref{ex20} and add $T$, $CX$ and $H$ gate:
   \[
    \begin{array}{c}  
      \Qcircuit @C=1em @R=.7em {
        \lstick{\ket{0}} & \qw\ar@{.}[]+<0em,1em>;[d]+<0em,-0.5em> & \gate{H} &\qw\ar@{.}[]+<0em,1em>;[d]+<0em,-0.5em> & \ctrl{1} & \qw\ar@{.}[]+<0em,1em>;[d]+<0em,-0.5em> & \gate{T} &\qw\ar@{.}[]+<0em,1em>;[d]+<0em,-0.5em> & \ctrl{1} & \qw\ar@{.}[]+<0em,1em>;[d]+<0em,-0.5em> &\gate{H} &\qw\ar@{.}[]+<0em,1em>;[d]+<0em,-0.5em> &\qw \\
        \lstick{\ket{0}} & \qw & \qw & \qw & \targ & \qw &\qw&\qw &\targ & \qw &\qw& \qw &\qw\\
        & \ket{\varphi_0} & & \ket{\varphi_1} & & \ket{\varphi_2} & & \ket{\varphi_3} & & \ket{\varphi_4} & & \ket{\varphi_5} &
      }
  \end{array}
  \]
Continuing from \autoref{ex20}, we obtain the following generalized generators.
\[
\begin{aligned}
    \underbrace{
    \gen{\begin{matrix}
    X\otimes X \\
    Z \otimes Z
    \end{matrix}}
    }_{\ket{\phi_2}}
    &\xrightarrow{T_0}
    \underbrace{
    \gen{\begin{matrix}
        X' \otimes X \\
        Z \otimes Z
      \end{matrix}} 
    }_{\ket{\phi_3}}
      \xrightarrow{CX_{0,1}}
    \underbrace{
        \gen{\begin{matrix}
      X' \otimes I \\
      I \otimes Z
      \end{matrix}}
    }_{\ket{\phi_4}}
    \xrightarrow{H_0}
    \underbrace{
        \gen{\begin{matrix}
      Z'\otimes I \\
      I \otimes Z
      \end{matrix}}
    }_{\ket{\phi_5}},
  \end{aligned}
\]
where $X' = \frac{1}{\sqrt{2}}(X+Y)$ and $Z' = \frac{1}{\sqrt{2}}(Z-Y)$.

In our encoding, as in the above definition of generalized stabilizer states, we let satisfying assignments represent not just the generator set, but the entire stabilizer groups.
These groups are:
  $ \mathcal{S}_{\ket{\phi_2}} = \{X \otimes X,Z \otimes Z,-Y \otimes Y,I \otimes I\}$,
    $\mathcal{S}_{\ket{\phi_3}} = \{X' \otimes X, Z \otimes Z, -Y' \otimes Y,I \otimes I\}$,
  $\mathcal{S}_{\ket{\phi_4}} = \{X'\otimes I,I \otimes Z, X'\otimes Z,I \otimes I\}$
  and $\mathcal{S}_{\ket{\phi_5}} = \{Z' \otimes I,I \otimes Z, Z'\otimes Z,I \otimes I\}$,
  where $Y' = \frac{1}{\sqrt{2}}(Y-X)$.

Finally, according to \autoref{eq:gendensity}, we may equally expand e.g.  $\mathcal{S}_{\ket{\phi_5}}$ to:\\
$\mathcal{S}_{\ket{\phi_5}}' ~~=~~ \{ \frac{1}{\sqrt{2}}  Z \otimes I,~~ -\frac{1}{\sqrt{2}}  Y \otimes I,~~I \otimes Z,~~ \frac{1}{\sqrt{2}}  Z \otimes Z,~~ -\frac{1}{\sqrt{2}}  Y \otimes Z,~~ I \otimes I\}$.
  \qed
\end{example}

\subsection{Encoding Clifford+T circuits}
Since generalized stabilizer states can be determined by a sum of weighted Pauli strings as shown in \autoref{eq:gendensity},
we encode a state by Boolean constraints whose satisfying assignments represent those weighted Pauli strings.
The idea is to encode the sign, the Pauli string, and the weights separately.
We will start with encoding the Pauli string and the sign.

A Pauli string $\mathbf{P}\in\mathrm{PAULI}_n$ can be encoded by $2n$ Boolean variables as
$2$ bits are needed for each of the $n$ Pauli matrices 
as $\sigma[x_{0},z_{0}] \otimes \ldots \otimes \sigma[x_{n-1},z_{n-1}]$,
where 
$j$-th Pauli matrix is indicated by variables $x_j$ and $z_j$.
To encode the sign,
only one Boolean variable $r$ is needed.
We introduce weighted model counting to interpret the sign by defining $W(r) = -1$ and $W(\no r) = 1$.
Thus $P\in \pm\mathrm{PAULI}_n$ can be interpreted as $(-1)^r\otimes_{i\in[n]}\sigma[x_i,z_i]$.
For example, consider Boolean formula $r\no x_1 z_1x_2z_2$. Its one satisfying assignment is $\{r\rightarrow 1, x_1 \rightarrow 0, z_1\rightarrow 1, x_2 \rightarrow 1, z_2\rightarrow 1\} \equiv -Z\otimes Y$.
Without loss of generality,
we set the initial state to be all zero state $\ket{0\dots 0}$, whose stabilizer group is $\mathcal{S}_{\ket{0\dots 0}} = \{\otimes_{i\in[n]} \Pbf_i \mid \Pbf_i\in\{Z,I\}\} \equiv \{(-1)^r\otimes_{i\in[n]}\sigma[x_i,z_i] \mid x_i = 0, z_i\in\{0,1\}\textup{ and } r = 0\}$.
Hence the Boolean formula for the initial state is defined as $F_{\init}(\bm{x}^0,\bm{z}^0,r) \defn \neg r \land \bigwedge_{j\in[n]}\neg x^0_{j}$,
where we use superscripts, e.g., $r^t, x_1^t, z_1^t$, to denote variables at time step $t$ (after $t$ gates from the circuit have been applied).
Note that we assign a weight to $r$ only on the final time step, 
as explained later.
\begin{example}
  Consider the initial state $\ket{00}$ in \autoref{ex35}.
  The corresponding constraint at time step $0$ is $\neg r^0\neg x^0_{0}\neg x^0_{1}$,
  which has satisfying assignments:
  \[
  \left\{
    \begin{aligned} 
    &\{r^0 \rightarrow 0, x_0^0 \rightarrow 0, x_1^0\rightarrow 0, z_0^0\rightarrow 1, z_1^0\rightarrow 1\} \\
    &\{r^0 \rightarrow 0, x_0^0 \rightarrow 0, x_1^0\rightarrow 0, z_0^0\rightarrow 1, z_1^0\rightarrow 0\} \\
    &\{r^0 \rightarrow 0, x_0^0 \rightarrow 0, x_1^0\rightarrow 0, z_0^0\rightarrow 0, z_1^0\rightarrow 1\} \\
    &\{r^0 \rightarrow 0, x_0^0 \rightarrow 0, x_1^0\rightarrow 0, z_0^0\rightarrow 0, z_1^0\rightarrow 0\}
    \end{aligned}
     \right\}  \equiv 
  \left\{
    \begin{aligned} 
    Z \otimes Z\\
    Z \otimes I\\
    I \otimes Z\\
    I \otimes I
    \end{aligned}
     \right\} 
  \]    
  \vspace{-4.5em}
  
  \hfill\qed
  \vspace{-.5em}
\end{example}
While we encode those signed Pauli strings using variables from $\set{x_{j}, z_{j}, r \mid j\in[n]}$,
to encode the weights, 
we introduce new variables $u$.
When a $T_j$ is performed and $x_{j}=1$, which means executing a $T$ gate on $j$-th qubit with certain stabilizer being either $\pm X$ or $\pm Y$,
we set $u = 1$ to indicate a branch of the operator, i.e. $TXT^\dag = X' = \frac{1}{\sqrt{2}}(X+Y)$ and $TYT^\dag= Y' =\frac{1}{\sqrt{2}}(Y-X)$.
Therefore, for each generalized stabilizer state in a circuit with $m$ gates and $n$ qubits, 
the encoding uses the set of variables
$V^t = \{x^t_{j}, z^t_{j}, r^t, u_{t}\mid \ j\in[n]\}$,
where $t\in\{0,\dots,m\}$ denotes a time step. 
\autoref{tab:booltgate} illustrates the details of how the Boolean variables in $V^t$ change over a $T$ gate. 
Here each satisfying assignment indicates a weighted Pauli string,
for example, there are two assignments for a stabilizer $\frac{1}{\sqrt{2}}(X+Y)$.

\begin{table}[t!]
\caption{Boolean variables under the action of conjugating one T gate. Here we omit the sign $(-1)^{r^{t}}$ for all $\gen{G}$ and sign $(-1)^{r^{t+1}}$ for all $\gen{TGT^\dag}$.}
\label{tab:booltgate}
\setlength{\tabcolsep}{6pt} %
\renewcommand{\arraystretch}{1} %
  \centering
  \begin{tabular}{rc|c|c|c|c|c}
    \hline
     $\gen G$  & $x^t z^t r^t$ & $\gen{TGT^\dag}$  & $x^{t+1}$ & $z^{t+1}$ & $r^{t+1}$ & $u$  \\
     \hline
     $I$    & 00 $r^t$ & $I$ & \multirow[]{2}{*}{ 0} & \multirow[]{2}{*}{$z^t$} & \multirow[]{2}{*}{$r^t$} &\multirow[]{2}{*}{0} \\
\cline{1-3}
    $Z$     & 01 $r^t$ &$Z$  &  & &  \\
    \hline
    $X$    & 10 $r^t$ &$\frac{1}{\sqrt{2}}(X+Y)$ & \multirow[]{2}{*}{1} & \multirow[]{2}{*}{\{0,1\}} & $r^t$ & \multirow[]{2}{*}{$1$} \\
    \cline{1-3}
    $Y$     & 11 $r^t$ &$\frac{1}{\sqrt{2}}(Y-X)$ & & & $r^t\oplus \neg z^{t+1}$ & \\
    \hline
  \end{tabular}
\end{table}

Using \autoref{tab:clifford} and \autoref{tab:booltgate},
given a single-qubit Clifford+T gate $\qg{j}$ on qubit~$j$ at time step $t$ (or $CX_{j,k}$ on qubits $j,k$),
we can derive a Boolean formula $F_{\qg_j}(V^t,V^{t+1})$,
abbreviated as $\qg_j^t$, as in the following.
\begin{equation}\label{cons:clifford}
  \begin{aligned}
    H^t_j \defn ~&r^{t+1} \Longleftrightarrow r^{t} \oplus x^t_{j}z^t_{j}
    ~~\land~~  z^{t+1}_{j} \Longleftrightarrow  x^t_{j}
    ~~\land~~  x^{t+1}_{j} \Longleftrightarrow  z^t_{j}\\
    S^t_j \defn  ~&r^{t+1} \Longleftrightarrow r^{t} \oplus x^t_{j}z^t_{j}
    ~~\land~~  z^{t+1}_{j} \Longleftrightarrow  x^t_{j} \oplus  z^t_{j}  \\
    CX^t_{j,k} \defn  ~& r^{t+1} \Longleftrightarrow r^{t} \oplus x^t_{j}z^t_{k}
      (x^t_{k} \oplus \neg z^t_{j})
    \
    ~~\land~~  x^{t+1}_{k} \Longleftrightarrow  x^t_{k} \oplus  x^t_{j} ~~\land~~ \\
    &  z^{t+1}_{j} \Longleftrightarrow  z^t_{j} \oplus  z^t_{k} \\
    T^t_j \defn  ~&x^{t+1}_{j} \Longleftrightarrow x^t_{j} \quad \wedge \quad 
  x^{t}_{j} \lor (z^{t+1}_{j} \Leftrightarrow z^t_{j}) \quad \wedge\\
  & r_i^{t+1} \Longleftrightarrow r^t \oplus  x^t_{j}  z^t_{j}  \neg z^{t+1}_{j}  \quad \wedge \quad u_{t} \Longleftrightarrow x^t_{j}.
  \end{aligned}
\end{equation}

The above omits additional constraints $v^{t+1} \Leftrightarrow v^t$
for all unconstrained time-step-$t+1$ variables, i.e., for all $v^t \in V_l^t$ with $l \neq j, k$.
In fact, it is not necessary to allocate new variables for those unconstrained time-step-$t+1$ variables.
The constraint $v^{t+1} \Leftrightarrow v^t$ can be effectively implemented by reusing the Boolean variables $v^t$ for $v^{t+1}$.
Thus for each time step, only a constant number of new variables need to be allocated.
For instance,
when applying $H_j$ gate,
we only need one new variable $r_i^{t+1}$,
since we can reuse the variable of $x_j^t$ for $z_j^{t+1}$ and  $z_j^t$ for $x_j^{t+1}$.
And when applying $CX_{ij}$ gate,
we need three new variables for $r^{t+1}$, $x_j^{t+1}$ and $z_j^{t+1}$.
Additionally, since variables for all $x_j^0$ and $z_j^0$ with $j\in[n]$ are allocated initially,
and as shown in \autoref{sec34},
performing a measurement introduces no new variable,
the size of our encoding is $O(n+m)$.

To this end, given a Clifford+$T$ circuit $C=\qg^0\cdots\qg^{m-1}$ without measurements,
we can build the following Boolean constraint.
\begin{equation}
\label{cons:clifft}
F_C(V^0,\dots,V^m)\defn F_\init(V^0) \wedge \bigwedge_{t\in[m]}F_{\qg^{t}}(V^t,V^{t+1}).
\end{equation}

The satisfying assignments of our encoding will represent weighted Pauli strings,
so that we can get the density operator at time $m$ by ranging over satisfying assignments $\alpha\in SAT(F_C)$:
  \begin{equation}
    \rho^m = \sum_{\alpha\in SAT(F_C)} F_C(\alpha) \cdot W(\alpha) \cdot\bigotimes_{j\in[n]}\sigma[\alpha(x^{m}_{j}),\alpha(z^{m}_{j})],
  \end{equation}
where $W(r^m) = -1$, $W(\no{r^m}) = 1$, $W(u_{t}) = \frac{1}{\sqrt{2}}$, $W(\no u_{t}) = 1$ for all $t\in\{0,\dots,m\}$
(and all other variables are unbiased).
So we will get the weight as $W(\alpha) = W(\alpha(r^m))\prod_{t\in[m]}W(\alpha(u_{t}))$. 
As mentioned before, we only assign weights to $r^m$ where $m$ is the final time step.
We allocate a new $r^{t+1}$ for each time step $t$ as we always get $r^{t+1}$ from a constraint related to $r^t$,
but the sign of the final state is given by $r^m$.
So we leave $r^t$ unbiased except when~$t$ is the final time step.
Additionally, there is no complex number assigned to any weights enabling the application of a classical weighted model counter that allows negative weights.
It is worth noting that instead of using satisfying assignments to represent stabilizers in Clifford circuits,
we represent weighted Pauli strings as satisfying assignments, 
which may be one summand of a stabilizer.
For Clifford circuits,
there are $2^n$ satisfying assignments for a $n$-qubit circuit at each time step.
While for Clifford+$T$ circuits,  
the satisfying assignments are more than $2^n$.

\begin{example}\label{ex36}
  Reconsider \autoref{ex35}, after solving the constraint $F_{\init} (V^0)\wedge H_0^0\wedge CX_{0,1}^1 \wedge T_0^2 \wedge CX_{0,1}^3 \wedge H_0^4$,
  the satisfying assignments encoding $\ket{\phi_5}$ will be 
  \[
    \left\{
        \begin{aligned}
            & \{r^5 \rightarrow 0, x^5_0 \rightarrow 0, x^5_1\rightarrow 0, z^5_0\rightarrow 1, z^5_1\rightarrow 0, u_2\rightarrow 1\}, \\
            & \{r^5\rightarrow 1, x^5_0 \rightarrow 1, x^5_1\rightarrow 0, z^5_0\rightarrow 1, z^5_1\rightarrow 0, u_2\rightarrow 1\}, \\
            & \{r^5 \rightarrow 0, x^5_0 \rightarrow 0, x^5_1\rightarrow 0, z^5_0\rightarrow 0, z^5_1\rightarrow 1, u_2\rightarrow 0\}, \\
            & \{r^5 \rightarrow 0, x^5_0 \rightarrow 0, x^5_1\rightarrow 0, z^5_0\rightarrow 1, z^5_1\rightarrow 1, u_2\rightarrow 1\}, \\
            & \{r^5 \rightarrow 1, x^5_0 \rightarrow 1, x^5_1\rightarrow 0, z^5_0\rightarrow 1, z^5_1\rightarrow 1, u_2\rightarrow 1\}, \\
            & \{r^5 \rightarrow 0, x^5_0\rightarrow 0, x^5_1\rightarrow 0, z^5_0\rightarrow 0, z^5_1\rightarrow 0, u_2\rightarrow 0\}
        \end{aligned}
      \right\}
      \equiv 
  \left\{
    \begin{aligned} 
    \tfrac{1}{\sqrt{2}} Z \otimes I\\
   -\tfrac{1}{\sqrt{2}} Y \otimes I\\
    I \otimes Z\\
    \tfrac{1}{\sqrt{2}} Z \otimes Z\\
   -\tfrac{1}{\sqrt{2}} Y \otimes Z\\
    I \otimes I
    \end{aligned}
     \right\} 
  \]
  where $w(\no r^5) = 1$, $w(\no r^5) = 1$, $w(u_2) = \frac{\sqrt{2}}{2}$ and $w(\no u_2) = 0$.
  Here we omit the satisfying assignments for $\{r^t,x^t_0,x^t_1,z^t_0,z^t_1 \mid 0\leq t\leq 4\}$. \qed
\end{example}

\subsection{Encoding random rotation gates}

Our encoding can be extended to other non-Clifford gates, which we demonstrate by adding rotation gates $RX(\theta)$, $RY(\theta)$, and $RZ(\theta)$,
where $\theta$ is an angle.

\def\horizontaldistance{\kern2pt}
\def\verticaldistance{10pt}

\begin{table}[t!]
    \caption{Lookup table for the action of conjugating Pauli gates by rotation gates. }
  \label{tab:rotation}
  \setlength{\tabcolsep}{15pt} 
  \centering
      \begin{tabularx}{\textwidth}{c | c | c c}
    \toprule
    \textbf{Gate} & \textbf{Matrix Form} & \textbf{In} & \textbf{Out} \\
    \midrule
      & \multirow{3}{*}{$\begin{bmatrix*}[r]
      \cos(\frac{\theta}{2}) & -i\sin(\frac{\theta}{2}) \\[\verticaldistance]
      -i\sin(\frac{\theta}{2}) & \cos(\frac{\theta}{2})
    \end{bmatrix*}$} & $X$ & $X$ \\
    $RX(\theta)$ & & $Y$ & $\cos(\theta)Y + \sin(\theta)Z$ \\
    &   & $Z$ & $\cos(\theta)Z - \sin(\theta)Y$ \\
    \midrule
    & \multirow{3}{*}{$\begin{bmatrix*}[r]
      \cos(\frac{\theta}{2}) & -\sin(\frac{\theta}{2}) \\[\verticaldistance]
      \sin(\frac{\theta}{2}) & \cos(\frac{\theta}{2})
    \end{bmatrix*}$} & $X$ & $\cos(\theta)Z + \sin(\theta)X$ \\
    $RY(\theta)$ & & $Y$ & $Y$ \\
    & & $Z$ & $\cos(\theta)X - \sin(\theta)Z$ \\
    \midrule
     & \multirow{3}{*}{    $\begin{bmatrix}
      \exp(-i\frac{\theta}{2}) & 0 \\[\verticaldistance]
      0 & \exp(i\frac{\theta}{2})
    \end{bmatrix}$} & $X$ & $\cos(\theta)X+\sin(\theta)Y$ \\
    $RZ(\theta)$ & & $Y$ & $\cos(\theta)Y - \sin(\theta)X$ \\
    & & $Z$ & $Z$ \\
    \bottomrule
  \end{tabularx}
\end{table}

In particular,
we have $T = \exp(-i\frac{\pi}{8})RZ(\frac{\pi}{4})$, $S = \exp(-i\frac{\pi}{4})RZ(\frac{\pi}{2})$ and $X = -iRX(\pi)$, $Y = -iRY(\pi)$, $Z = -iRZ(\pi)$.
Note however that the stabilizer formalism discards the global phase of a state as it updates stabilizers by conjugation, e.g., 
$TPT^\dag = \left(\exp(-i\frac{\pi}{8})RZ(\frac{\pi}{4})\right) P \left(\exp(-i\frac{\pi}{8})RZ(\frac{\pi}{4})\right)^\dag = RZ(\frac{\pi}{4}) P RZ(\frac{\pi}{4})^\dag $.
Based on \autoref{tab:rotation},
the constraints for rotation gates are as below, 
where we keep the coefficients $\cos(\theta)$ and $\sin(\theta)$ by defining the weights of the new variables 
as $w(u_{1t}) = \cos(\theta)$, $w(\neg u_1) = w(\neg u_2) = 1$ and $w(u_2) = \sin(\theta)$.
\begin{align*}
 & RX^t_j \defn 
 &&z^{t+1}_{j} \Longleftrightarrow z^t_{j} ~~\land~~
  z^{t}_{j} \vee (x^{t+1}_{j} \Leftrightarrow x^t_{j})~~\land~~
 r^{t+1} \Longleftrightarrow r^t \oplus  z^t_{j} \neg x^t_{j} x^{t+1}_{j} ~~\land~~ \\
 & &&u_{1t} \Longleftrightarrow z^t_{j}(x^t_{j}x^{t+1}_{j}\vee \neg x^t_{j} \neg x^{t+1}_{j}) ~~\land~~ 
 u_{2t} \Longleftrightarrow z^t_{j}(\neg x^t_{j}x^{t+1}_{j}\vee x^t_{j} \neg x^{t+1}_{j}). \\
 & RY^t_j \defn 
 && (x^{t}_{j} \oplus z^{t}_{j}) \Longleftrightarrow (x^{t+1}_{j} \oplus z^{t+1}_{j}) ~~\land~~ 
 (x^{t}_{j} \oplus \neg z^{t}_{j}) \Longrightarrow (x^{t}_{j} z^{t}_{j} \Leftrightarrow x^{t+1}_{j} \Leftrightarrow z^{t+1}_{j})
 \land~~ \\
 & && r^{t+1} \Longleftrightarrow r^t \oplus  z^t_{j} z^{t+1}_{j} ~~\land~~ u_{1t} \Longleftrightarrow ( x^{t+1}_{j}z^{t}_{j}\oplus x^{t}_{j}z^{t+1}_{j}) ~~\land~~ \\
 & &&u_{2t} \Longleftrightarrow (x^{t+1}_{j}x^{t}_{j}\oplus z^{t+1}_{j}z^{t}_{j}). \\
 & RZ^t_j \defn 
 &&x^{t+1}_{j} \Longleftrightarrow x^t_{j} ~~\land~~
 x^{t}_{j} \lor (z^{t+1}_{j} \Leftrightarrow z^t_{j})~~\land~~
 r^{t+1} \Longleftrightarrow r^t \oplus  x^t_{j}  z^t_{j}  \neg z^{t+1}_{j} ~~\land~~ \\
 & &&u_{1t} \Longleftrightarrow x^t_{j}(z^t_{j}z^{t+1}_{j}\vee \neg z^t_{j} \neg z^{t+1}_{j}) ~~\land~~ 
 u_{2t} \Longleftrightarrow x^t_{j}(\neg z^t_{j}z^{t+1}_{j}\vee z^t_{j} \neg z^{t+1}_{j}).
\end{align*}

\subsection{Measurement}\label{sec34}

In this section,
we consider projective measurement both on a single qubit and on multiple qubits of a quantum system.
Single-qubit measurement~\cite{classical2021kocia} can be used for extracting only one bit of information from a $n$-qubits quantum state to effectively protect quantum information~\cite{optimizing2023polla}.
It is also used in random quantum circuits, which contributes to the study of quantum many-body physics~\cite{random2023fisher}.
Measurement on multiple qubits is generally used in quantum algorithms, such as in Grover and  Shor algorithms, to get the final result.
We implement both measurements using \emph{Pauli measurement},
where projectors are Pauli strings.

\vspace{-1em}
\subsubsection{Single-qubit Pauli Measurement.}
Let $\mathbb{P}_{k,0} = I \otimes \cdots \otimes \op{0}{0}_k \otimes \cdots \otimes I= \frac{1}{2}(Z_k + I^{\otimes n})$ and 
$\mathbb{P}_{k,1} = I \otimes \cdots \otimes \op{1}{1}_k \otimes \cdots \otimes I
    = \frac{1}{2}(- Z_k + I^{\otimes n})$ for $k\in[n]$.
When measuring $k$-th qubit of a $n$-qubit state $\ket{\psi}$ using projectors $\{\mathbb{P}_{k,0}, \mathbb{P}_{k,1}\}$,
we get two possible outcomes: $0$ with probability $p_{k,0}$  and $1$ with probability $p_{k,1}$.
It follows that $p_{k,0} = \Tr(\mathbb{P}_{k,0}\op{\psi}{\psi})$, where $\tr$ is the trace mapping~\cite{nielsen2000quantum}.
As shown in \autoref{eq:gendensity},
the density operator $\op{\psi}{\psi}$ can be written using generalized stabilizers i.e., $\op{\psi}{\psi} = 
\frac1{2^n} \sum_{P\in\mathcal{S}_{\ket{\psi}}} P$ for 
 $P = \lambda_{P} \mathbf{P}$ where $\mathbf{P}\in\mathrm{PAULI}_n$ and $\lambda_{P}\in\mathbb{R}$. 
The probability $p_{k,0}$ can be obtained as follows.
\begin{equation}\label{eq:prob}
  \begin{aligned}
    p_{k,0} & = \tr(\mathbb{P}_{k,0}\op{\psi}{\psi}) = \tr(\frac{1}{2}( I^{\otimes n} + Z_k)\op{\psi}{\psi})   \\ 
    & = \frac{1}{2}(\tr(\op{\psi}{\psi}) + \tr(Z_k \op{\psi}{\psi})) = \frac{1}{2} + \frac{1}{2^{n+1}}\sum_{\lambda_{P}\mathbf{P}\in\mathcal{S}_{\ket{\psi}}}\lambda_{P}\tr(Z_k \mathbf{P}) 
  \end{aligned}
\end{equation}
For any $\mathbf{P}\in\mathrm{PAULI}_n$,
the trace $\tr(\mathbf{P})$ is non-zero if and only if $\mathbf{P} = I^{\otimes n}$.
Given the fact that $\tr(A\otimes B) = \tr(A)\tr(B)$,
when considering a $n$-Pauli string $\Pbf= \Pbf_0\otimes\dots\otimes \Pbf_{n-1}$,
we can express $\tr(\Pbf)$ as a product of its constituent Pauli matrices $\tr(\Pbf_0)\cdots\tr(\Pbf_{n-1})$.
Since it is easy to see that $\tr(X) = \tr(Y) = \tr(Z) = 0$,
if there exists a $\Pbf_i\in\{X,Y,Z\}$, then $\tr(\Pbf) = 0$.
Furthermore, $\Pbf = Z_k$ if and only if
$Z_k\Pbf = I^{\otimes n}$.
Thus the result of \autoref{eq:prob} can be simplified as
$p_k = \frac{1}{2} + \frac{1}{2}\sum_{P = \lambda_P Z_k} \lambda_P$.
In other words, to get the probability of the outcome being 0 on measuring $k$-th qubit,
we need to sum up the weights of all elements $Z_k$ in the flattened Pauli group of \autoref{eq:gendensity}. The latter can be encoded as follows, where $m$ is the final time step.
\begin{equation}\label{cons:measure}
  F_{M_k}(V^m) \defn \bigwedge_{j\in[n]} \neg x^m_j \wedge \bigwedge_{j\in[n], j\neq k} \neg z^m_j \wedge z_k^m
\end{equation}
A Clifford+$T$ circuit with a single-qubit Pauli measurement at the end can be encoded by conjoining the constraint of initial state and gates in \autoref{cons:clifft}, and the one for the measurement at the end.
\begin{example}
Consider the circuit in \autoref{ex36} and assume we want to perform single-qubit Pauli measurement on the first qubit using $\{\frac{I\otimes I+Z_0}{2}, \frac{I\otimes I-Z_0}{2}\}$.
\[
  \begin{array}{c}  
    \Qcircuit @C=1em @R=.7em {
      \lstick{\ket{0}} & \gate{H} & \ctrl{1} & \gate{T} & \ctrl{1} & \gate{H} & \meter \\ 
      \lstick{\ket{0}} & \qw & \targ & \qw & \targ & \qw & \qw
    }
\end{array}
\]
By adding the measurement constraint $  M_0 \defn \no x^5_0 \wedge \no x^5_1 \wedge z^5_0 \wedge \no z^5_1$ to the circuits constraints in \autoref{ex36},
we get $F_{\init}(V^0)\wedge H_0^0\wedge CX^1_{0,1} \wedge T_0^2 \wedge CX_{0,1}^3\wedge H_0^4 \wedge F_{M_0}(V^5)$.
The satisfying assignments will be the subset of the solutions in \autoref{ex36}:
$\left\{\sigma = \{r^5 \rightarrow 0, x^5_0 \rightarrow 0, x^5_1\rightarrow 0, z^5_0\rightarrow 1, z^5_1\rightarrow 0, u_2\rightarrow 1\}\right\}$,
where we only show the variables in $V^5$ representing the final state.
The resulting probability is
$W(\sigma(r^5))W(\sigma(u_2)) = \frac{1}{\sqrt{2}}$.
\qed
\end{example}

\vspace{-1em}
\subsubsection{Multi-qubit Pauli Measurement.}

Similar to the single-qubit Pauli measurement,
we can resolve the constraint of multi-qubit Pauli measurement based on the measurement projector $\mathbb{P}$.
Let $Q$ be the set of all qubits and $\Bar{Q}\subseteq Q$ be the set of qubits being measured.
The projector measuring qubits in $\Bar{Q}$ is defined as $\mathbb{P}_{\Bar{Q}} = \bigotimes_{q\in Q}\mathbb{P}_q$ where $\mathbb{P}_q=(I+Z)/2$ for $q\in\Bar{Q}$ and $\mathbb{P}_q=I$ for $q\in Q\backslash\Bar{Q}$.
We can derive the constraint $F_{\mathbb{P}_q}(x_q^m,z_q^m) = \no x_q^m$ for $q\in\Bar{Q}$ and no constraint for $q\in Q\backslash\Bar{Q}$.
Thus $F_{M_{\Bar{Q}}} = \bigwedge_{q\in\Bar{Q}}\no x_q^m$.
In this way, 
it is possible to construct the constraint for any computational basis measurement. 

We conclude \autoref{sec:encoding} with \autoref{prop:simu}, which also shows that our encoding implements a strong simulation of a universal quantum circuit.
\begin{proposition}\label{prop:simu}
Given an $n$-qubit quantum circuit $C$, its encoding 
$F(V^0,\dots,V^m) = F_C(V^0,\dots,V^m)\wedge F_{M_{\Bar{Q}}}(V^m)$ with according weight function $W$ and a computational basis projector $\mathbb{P}$ on $q \leq n$ qubits ($q = \sizeof{\bar Q}$), WMC on $F$ computes the probability of measuring the outcome corresponding to $\mathbb{P}$ on circuit $C$, i.e., we have 
$\frac{1}{2^{q}}MC_W(F) = \bra{0}^{\otimes n}C^\dag\mathbb{P} C\ket{0}^{\otimes n}$.
\end{proposition}

\section{Experiments}
\label{S5}
\label{sec:experiments}

To show the effectiveness of our approach, we implemented a WMC-based simulator in a tool called QCMC.
It reads quantum circuits in QASM format~\cite{cross2022openqasm}, encodes them
to Boolean formulas in conjunctive normal form (CNF) as explained in \autoref{sec3},
and then uses the weighted model counter GPMC~\cite{hashimoto2020gpmc} to solve these constraints.
We choose GPMC as it is the best solver supporting negative weights in model counting competition 2023~\cite{modelcounting2023}.
The resulting implementation and evaluation are publicly available at 
\href{https://github.com/System-Verification-Lab/Quokka-Sharp}{https://github.com/System-Verification-Lab/Quokka-Sharp}.

We performed a classical simulation of a quantum circuit 
comparing our method  against two state-of-the-art tools: QuiZX~\cite{kissinger2022simulating} based on ZX-calculus~\cite{coecke2011interactingZXAlgebra} and Quasimodo~\cite{quasimodo} based on CFLOBDD~\cite{sistla2023weighted}.
In particular, this empirical analysis is performed on two families of circuits:
(i) random Clifford+$T$ circuits, which mimic hard problems arising in quantum chemistry~\cite{wright2022chemistry} and quantum many-body physics~\cite{random2023fisher};
(ii) random circuits mimicking oracle implementations; %
(iii) all benchmarks from the public benchmark suite MQT Bench~\cite{mqt2023quetschlich},
which includes many important quantum algorithms like QAOA, VQE, QNN, Grover, etc.
All experiments have been conducted on a 3.5 GHz M2 Machine with MacOS 13 and 16 GB RAM. 
We set the time limit to be 5 minutes (300 seconds) and include the time to read a QASM file, construct the weighted CNF, and perform the model counting in all reported runtimes.

\vspace{-1em}
\paragraph{Results.}

First, we show the limit of three methods using the set of benchmarks generated by \cite{kissinger2022simulating}.
They construct random circuits with a given number of $T$ gates by exponentiating Pauli unitaries in the form of $\exp(-i(2k+1)\frac{\pi}{4}P)$ where $P$ is a Pauli string and $k\in\{1,2\}$.
We reuse their experimental settings, which gradually increase the $T$ count (through Pauli exponentiation) for $n =50$ qubits, and we add an experiment with $n = 100$ qubits.
Accordingly, we generate 50 circuits with different random seeds for each $n = 50$ and  $T \in [0\text{--}100]$ and each $n = 100$ and  $T \in [0\text{--}180]$. We then perform a single-qubit Pauli measurement on the first qubit.
We plot the minimal time needed and the rate of successfully getting the answer in 5 minutes among all 50 simulation runs
in \autoref{fig:gadget}.

\begin{figure}[h!]
   \vspace*{-5mm} 
  \centering
  \includegraphics[height=3in]{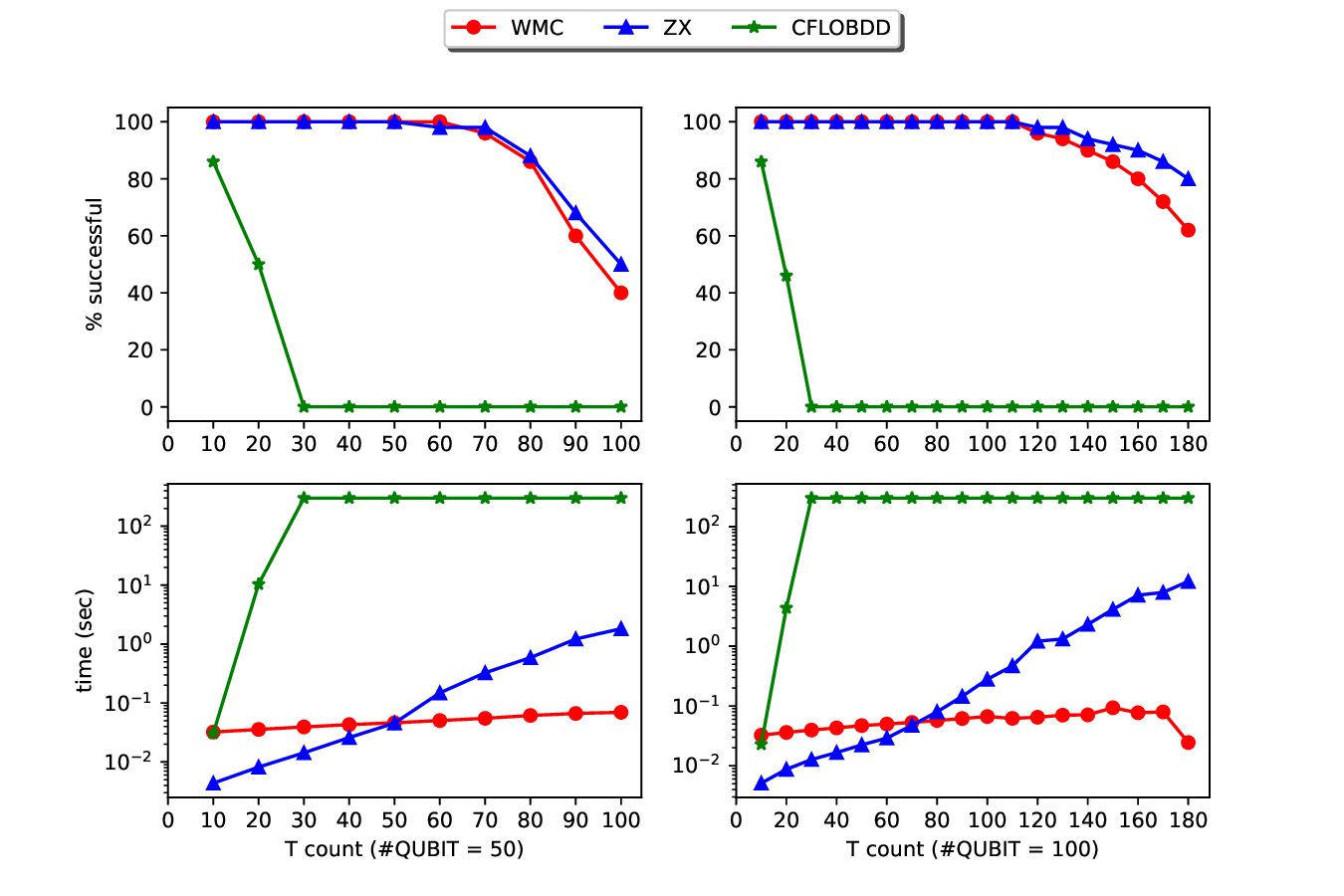}
  \caption{The upper two figures, both of which have y-axes on a logarithmic scale, are percentages of random 50- and 100-qubit circuits with increasing depth which can be successfully measured in 5 minutes. The below two figures are the minimum running time among the 50 samples for each configuration.}
  \label{fig:gadget}
\end{figure}

Second, we also consider random circuits that more resemble typical oracle implementations --- random quantum circuits with varying qubits and depths, which comprise the $CX$, $H$, $S$, and $T$ gates with appearing ratio 10\%, 35\%,  35\%,  20\%~\cite{peham2022equivalence}.
The resulting runtimes can be seen in \autoref{fig:scale}.

\begin{figure*}[t!]
  \centering
  \begin{subfigure}[t]{0.5\textwidth}
      \centering
      \includegraphics[height=1.8in]{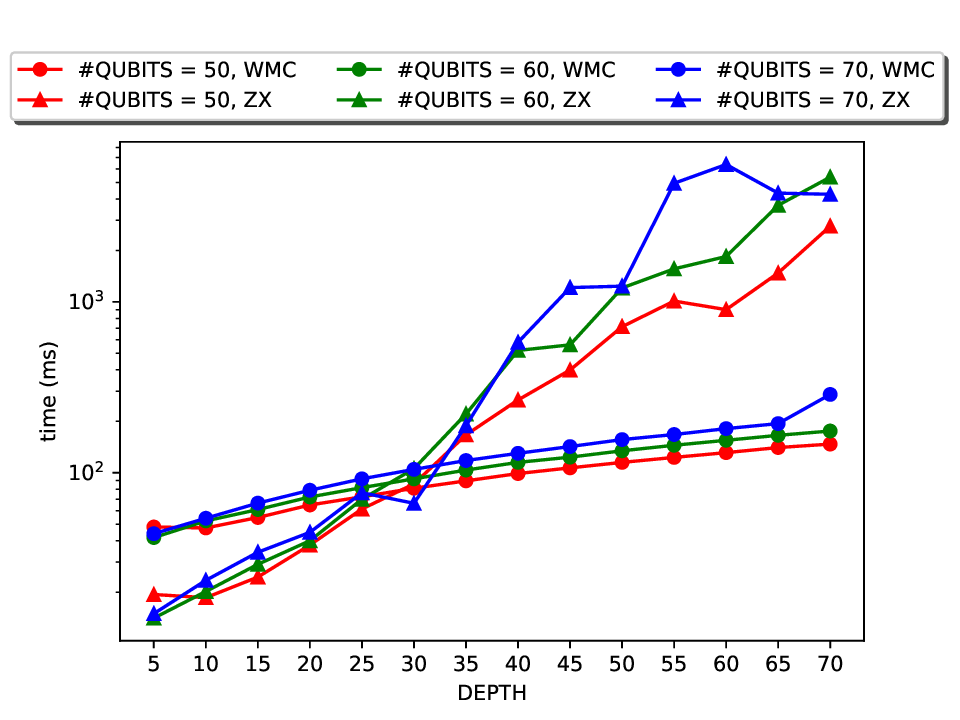}
      \caption{Runtimes for growing qubit counts}
  \end{subfigure}%
  ~ 
  \begin{subfigure}[t]{0.5\textwidth}
      \centering
      
      \includegraphics[height=1.8in]{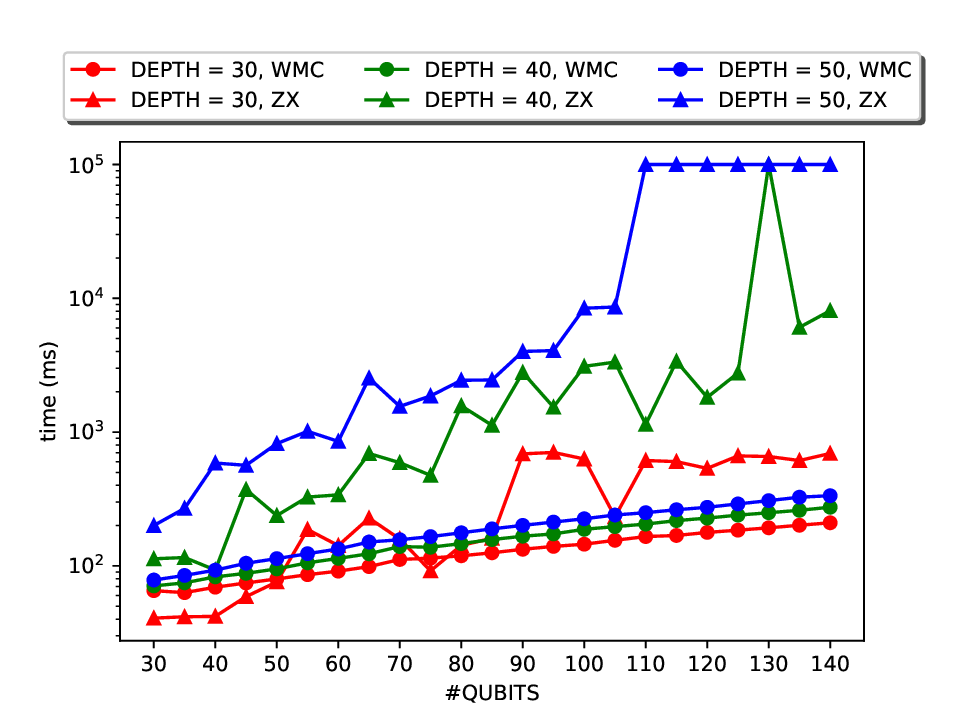}
      \caption{Runtimes for growing circuit depth}
  \end{subfigure}
   \vspace*{-1em} 
  \caption{Computational basis measurement of typical random Clifford+T circuits. (Both vertical axes are on a logarithmic scale.)
  CFLOBDD runs out of time for all benchmarks so we do not add it here.
  }
  \label{fig:scale}
\end{figure*}

In addition to random circuits, 
we empirically evaluated our method on the MQTBench benchmark set~\cite{mqt2023quetschlich}, measuring all qubits, as is typical in most quantum algorithms.
We present a representative subset of results in \autoref{tab:alg}.
The complete results can be found in \hyperref[appendix]{Appendix~\ref*{appendix}}. %
All benchmarks are expanded to the Clifford+T+R gate set,
where $R$ denotes
$\{RX,RY,RZ\}$.
The first two columns list the number of qubits $n$ and the number of gates $|G|$. 
Columns $T$ and $R$ give the number of $T$ gates and rotation gates. 
Then, the performance of the weighted model counting tool QCMC, 
the performance of the ZX-calculus tool QuiZX (ZX),
and the performance of CFLOBDD tool Quasimodo (CFLOBDD).
The performances are given by the runtime and the corresponding memory usage.

\newcommand\ERROR{crash} 
\newcommand\timeout{$>300$}

\begin{table}[h]
  \setlength{\tabcolsep}{1.5pt} %
     \caption{Results of verifying circuits from MQT bench. 
     For cases within time limit,
     we give their running time (sec) and corresponding memory usage (MB).
     We use \ding{53} \ when QuiZX does not support certain benchmarks, while $>300$ represents a timeout (5 min).
     For those benchmarks having a timeout or are not supported,
     we omit their memory usage by $\star$.
     }
     \label{tab:alg}
    \centering
    \small
    \begin{tabular}{ c | r r r r | c | c | c | c | c | c}
      \hline
      \multirowcell{2}{Algorithm} & \multirowcell{2}{n} & \multirowcell{2}{$|G|$} & \multirowcell{2}{$T$} & \multirowcell{2}{$R$} & \multicolumn{2}{c|}{WMC} &  \multicolumn{2}{c|}{ZX}  &  \multicolumn{2}{c}{CFLOBDD}  \\
      \cline{6-11}
      &  & & & & t(sec) & RSS(MB) &  t(sec) & RSS(MB)  &  t(sec) & RSS(MB)  \\
      \hline
      \multirowcell{3}{GHZ\\State}        
      & 32 & 32 & 0 & 0 & 0.044 & 12.56 & 0.007 & 12.28 &  0.03 & 354.81\\
      & 64 & 64 & 0 & 0 & 0.049 & 12.54 & 0.008 & 12.30 &  0.03 & 356.02\\
      & 128 & 128 & 0 & 0 & 0.048 & 12.56 & 0.013 & 12.44 &  0.04 & 355.77\\
      \hline
  
      \multirowcell{3}{Graph\\State}        
      & 16 & 64 & 0 & 0 & 0.046 & 12.53 & 0.005 & 12.44 & 0.06 & 355.68 \\
      & 32 & 128 & 0 & 0 & 0.045 & 12.53 & 0.008 & 12.40 & 0.11 & 355.39 \\
      & 64 & 256 & 0 & 0 & 0.045 & 12.36 & 0.015 & 12.47 & 243.22 & 6116.93\\
      \hline
  
      \multirowcell{3}{Grover's\\(no ancilla)}       
      & 4 & 162 & 8 & 58 & 0.18 & 12.53  & 89.18 & 12.06  & 0.04 & 345.56 \\
      & 5 & 470 & 0 & 195 & 10.18 & 12.44 & \timeout & $\star$ & 0.07 & 345.73  \\
      & 6 & 1314 & 0 & 552 & \timeout & $\star$ & \timeout & $\star$ & 0.21 & 346.52 \\
      \hline
      
      \multirowcell{3}{QAOA}       
      & 7 & 63 & 0 & 28 & 0.03 & 12.42 & \multirowcell{3}{\timeout} & \multirowcell{3}{$\star$} & 0.05 & 355.98 \\
      & 9 & 81 & 0 & 36 & 0.036 & 12.28 &  &  & 0.05 & 355.89 \\
      & 11 & 99 & 0 & 44 & 0.035 & 12.41 &  &  & 0.06 & 355.42 \\
      \hline
  
      \multirowcell{3}{QNN}       
      & 16 & 1119 & 0 & 400 & \multirowcell{3}{\timeout} & \multirowcell{3}{$\star$} & \multirowcell{3}{\ding{53}} & \multirowcell{3}{$\star$} & 57.36 & 2232.14 \\
      & 32 & 3775 & 0 & 1312 & & & & & \timeout & $\star$ \\
      & 64 & 13695 & 0 & 4672 & & & & & \timeout & $\star$ \\
      \hline
      
      \multirowcell{3}{Quantum\\ Walk \\ (no ancilla)} 
      & 5 &  1071 & 24 & 448 & 142.75 & 12.70 & \multirowcell{3}{\ding{53}} & \multirowcell{3}{$\star$}  & 0.13 & 345.19 \\
      & 6 & 2043  & 24 &  844 &  \timeout & $\star$ & & & 0.28 & 345.58 \\
      & 7 & 3975  & 24 & 1624 & \timeout & $\star$ & & & 0.80 & 347.58 \\   
      \hline
      
      \multirow{3}{*}{QFT} 
            & 16 & 520 & 0 & 225 & 0.05 & 12.34 & 0.09 & 12.30 & 6.58 & 516.19\\
            & 32 & 2064 & 0 & 961 & 0.11 & 12.34 & 0.16 & 12.53 & \timeout & $\star$ \\
            & 64 & 8224 & 0 & 3969 & 0.37 & 12.44 & 0.57 & 12.32 & \timeout & $\star$ \\ 
      \hline
      \multirowcell{3}{VQE}        
      & 14 & 236 & 0 & 82 & 0.23 & 12.28 & \multirowcell{3}{\ding{53}} & \multirowcell{3}{$\star$}  & 2.84 & 610.54 \\
      & 15 & 253 & 0 & 88 & 0.49 & 12.14 & & & 6.29 & 680.64\\
      & 16 & 270 & 0 & 94 & 0.51 & 12.50 & & & 17.97 & 943.67 \\
      \hline
      \multirowcell{3}{W-state}        
      & 32 & 435 & 0 & 124 & 0.11 & 12.59 & \multirowcell{3}{\ding{53}} & \multirowcell{3}{$\star$} & \multirowcell{3}{\timeout} & \multirowcell{3}{$\star$} \\
      & 64 & 883 & 0 & 252 & 0.28 & 12.62 & & &  & \\
      & 128 & 1779 & 0 & 508 & 0.66 & 12.63 & & &  &  \\
      \hline
     \end{tabular}
      \vspace*{-2.5em}
  \end{table}

\vspace{-1em}
\paragraph{Discussion.}

For all cases, QuiZX gives algebraic answers while the QCMC and CFLOBDD methods give numerical answers.
Because of the imprecision of floating-point arithmetic,
we consider $a$ equal to $b$ if $|a-b|<10^{-8}$.
With equality tolerance,
all three methods produce the same answers.

For random circuits,
\autoref{fig:gadget} illustrates that the minimum runtime barely increases for QCMC, while it seems exponential for QuiZX (note the log scale).
However,  when the number of qubits  is $n = 50$ and the $T$ count is larger than $70$,
or when, $n = 100$ and the $T$ count is larger than 110,
QuiZX has a better success rate, i.e., it completes more simulations than QCMC in 5 minutes.
In contrast, CFLOBDD exhibits the lowest success rate among the three methods. 
When it comes to a typical random Clifford+T circuit 
\autoref{fig:scale} shows that the runtime of both QCMC and ZX exhibits a clear correlation with the size of the circuits,
while CFLOBDD can not solve all benchmarks in 5 minutes.
The proposed implementation consistently outperforms QuiZX by one to three orders of magnitude especially when the size is getting larger (note again the log scale).
However, the story changes when considering structural quantum circuits.

For MQT benchmarks,
\autoref{tab:alg} shows that QCMC performs better than QuiZX except for GHZ state and Graph State where QCMC is
slightly slower in milliseconds.
CFLOBDD significantly surpasses QCMC on Grover and quantum walk algorithms,
primarily due to the decision diagram-based method's proficiency in handling circuits featuring large reversible parts and oracles.
While for those circuits featuring a large number of rotation gates with various rotation angles, 
like Graph state, QFT, and VQE, QCMC demonstrates clear advantages. 
This distinction arises from the fact that when dealing with rotation gates,
it might happen that two decision diagram nodes that should be identical in theory, differ by a small margin in practice, obstructing node merging~\cite{9023381}. 
In contrast, the WMC approach ---also numerical in nature--- avoids explicit representation of all satisfying assignments, by iteratively computing a sum of products. This not only avoids blowups in space use but, we hypothesize, also avoids numerical instability, a problem that has plagued numerical decision-diagram based approaches~\cite{peham2022equivalence,9023381}. 
In terms of memory usage,
CFLOBDD always uses more than 340 MB,
in some cases uses  more than 6 GB (graph state, $n = 64$),
while QCMC and QuiZX use less than 13 MB (OS reported peak resident set size).

Overall,  both QCMC and QuiZX outperform CFLOBDD in handling random circuits.
Moreover, QCMC has better runtime performance than QuiZX.
For structural circuits,
QuiZX faces a limitation as it does not efficiently support rotation gates with arbitrary angles, 
so it is incapable of simulating many quantum algorithms, like VQE, directly.
In terms of runtime,
CFLOBDD is better at circuits featuring structure,
while QCMC performs better at circuits with arbitrary rotation gates.
However, CFLOBDD has a significantly higher memory cost compared to both QCMC and QuiZX.
 
\section{Related Work}
\label{sec:related}

In this section, we give an overview of the related work on classical simulation of quantum computing with a focus on those methods applying SAT-based solvers.

SAT-based solvers have proven successful in navigating the huge search spaces encountered in various problems in quantum computing~\cite{robert2011atpg,meuli2018sat},%
initial attempts have been made to harness the strengths of satisfiability solvers for the simulation of quantum circuits.
For instance,
\cite{berent2022towards} implements a simulator for Clifford circuits based on a SAT encoding (our encoding of $H, S, CX$ in \autoref{cons:clifford} is similar to theirs).
The authors also discuss a SAT encoding for universal quantum circuits, which however requires exponentially large representations, making it impractical.
Besides quantum circuits, \cite{Bauer2023symQV} presents symQV, a framework for verifying quantum programs in a quantum circuit model,
which allows to encode the verification problem in an SMT formula,
which can then be checked with a $\delta$-complete decision process. There is an SMT theory for quantum computing~\cite{chen2023theory}.

Another method is based on decision diagrams (DDs)~\cite{akers,bryant86},
which represent many Boolean functions succinctly, while allowing manipulation operations without decompression.
DD methods for pseudo-Boolean functions include Algebraic DDs (ADD)~\cite{bahar1993algebraic,clarke1993spectral,viamontes2003improving} and various ``edge-valued'' ADDs~\cite{lai1994evbdd,tafertshofer1997factored,wilson2005decision,sanner2005AffineADDs}.
The application of DDs to quantum circuit simulation,
by viewing a quantum state as a pseudo-Boolean function, was pioneered with QuiDDs~\cite{viamontes2004high}
and further developed with Quantum Multi-valued DDs~\cite{miller2006qmdd}, Tensor DDs~\cite{hong2021tensor} and CFLOBDDs~\cite{sistla2023weighted}. %
All but CFLOBDD are essentially ADDs with complex numbers.

Another way is to translate quantum circuits into ZX-diagrams~\cite{coecke2011interactingZXAlgebra}, 
which is a graphical calculus for quantum circuits equipped with powerful rewrite rules.

Classical simulation is commonly used for the verification of quantum circuits, with extensive research focused on their equivalence checking~\cite{wang2022equivalence,al2014verification,al2013equivalence}. 
It can also be applied to bug hunting in quantum circuits.
In \cite{chen2023automata},
the authors proposed a tree automata to compactly  represent quantum states and gates algebraically, framing the verification problem as a Hoare triple.

\section{Conclusions}

In this work, we propose a generalized stabilizer formalism formulated in terms of a stabilizer group.
Based on this, we provide an encoding for various universal gate sets as a weighted model counting problem with only negative weights, obviating the need for complex numbers that are not supported by existing WMC tools. 
Besides $T$ gates,
we also extend our encoding to general rotation gates.
Furthermore, we demonstrate how to perform computational basis measurements using this encoding, enabling strong quantum circuit simulation.

We have implemented our method
in an open source tool QCMC.
To give empirical results on the practicality of our method,
we applied it to a variety of benchmarks comparing one based on ZX-calculus and one based on decision diagrams.
Experimental results show that our approach outperforms it in most cases, particularly with circuits of large sizes.
The performance of our approach is quite different from the other approaches, demonstrating the unique potential of WMC in various use cases. It will be interesting to see whether WMC tools can improve up on this should these new benchmarks be included in the WMC competition.

This work provides a new benchmark paradigm for weighted model counting problems.
In the future, it will be interesting to apply our encoding to approximate weighted model counting~\cite{chakraborty2014distribution,ermon2013embed}.
Moreover, using weighted samplers~\cite{meel2022inc,golia2021designing}, we could realize weak circuit simulation using the same encoding.
The main obstacle now is to allow negative weights in approximate weighted model counters and samplers.
Additionally, it will also be interesting to explore more applications of this encoding, such as checking the equivalence of two quantum circuits and entanglement purification.

\newpage
\bibliographystyle{plain}
\bibliography{lit} 

\newpage

\appendix
\section{Appendix}\label{appendix}
Here we show the results of remaining benchmarks in MQT bench,
where \autoref{tab:fullres} and \autoref{tab:fullres2} shows results of simulating measurements on scalable benchmarks and non-scalable benchmarks respectively~\cite{mqt2023quetschlich}.
 
  \begin{table}[!h]
    \caption{Results of verifying scalable Benchmarks in MQTbench.}
    \label{tab:fullres}
    \setlength{\tabcolsep}{1pt} 
    \scriptsize
    \begin{tabular}{ c | r r r r | c | c | c | c | c | c}
      \hline
      \multirowcell{2}{Algorithm} & \multirowcell{2}{n} & \multirowcell{2}{$|G|$} & \multirowcell{2}{$T$} & \multirowcell{2}{$R$} & \multicolumn{2}{c|}{WMC} &  \multicolumn{2}{c|}{ZX}  &  \multicolumn{2}{c}{CFLOBDD}  \\
      \cline{6-11}
      &  & & & & t(sec) & RSS(MB) &  t(sec) & RSS(MB)  &  t(sec) & RSS(MB)  \\
      \hline
        \multirowcell{4}{Amplitude\\Estimation}        
        & 16 & 725 & 0 & 390 & \multirowcell{4}{\timeout} & \multirowcell{4}{$\star$}  & \timeout & \multirowcell{4}{$\star$} & \multirowcell{4}{\timeout} & \multirowcell{4}{$\star$}\\
        & 32 & 2485 & 0 & 1297 & & & \timeout &  & &  \\
        & 64 & 9077 & 0 & 4642 & & & \timeout & & & \\
        & 128 & 34549 & 0 & 17467 & & & \ERROR & & & \\
        \hline
      \multirowcell{3}{Deutsch\\Jozsa}        
      & 16 & 160 & 0 & 56 & 0.05 & 12.56 & 0.008 & 12.45  & 0.04 & 355.92 \\
      & 32 & 470 & 0 & 195 & 0.05 & 12.59 & 0.012 & 12.69 &  0.04 & 355.43\\
      & 64 & 1314 & 0 & 552 & 0.07 & 12.55 & 0.022 & 12.44 & 0.07 & 355.54\\
      & 128 & 3319 & 0 & 1344 & 0.1 & 12.67 & 0.04 & 12.47 & 0.1 & 356.05\\

      \hline 
        \multirowcell{4}{Grover's\\(v-chain)}       
        & 5 & 458 & 24 & 165 & 23.41 & 12.75 & 0.05 & 12.43 & 0.07 & 355.91 \\
        & 7 & 1058 & 48 & 388 & \timeout & $\star$ & 0.28 & 12.53 & 0.11 & 355.72 \\
        & 9 & 1314 & 0 & 552 & \timeout & $\star$ & 1.19 & 12.48 & 0.32 & 365.06\\
        & 11 & 3319 & 0 & 1344 & \timeout & $\star$ & 4.86 & 12.43 & 1.10 & 386.73 \\
        \hline
    
        \multirowcell{3}{Portfolio \\ Optimization \\ with QAOA}       
        & 7 & 315 & 0 & 161 & 41.17 & 12.64 & \multirowcell{3}{\ding{53}} & \multirowcell{3}{$\star$}  & 0.12 & 355.64 \\
        & 9 & 470 & 0 & 195 & \timeout & $\star$ & & &  0.32 & 360.84\\
        & 11 & 1314 & 0 & 552 & \timeout & $\star$ & & & 1.20 & 370.64\\
        \hline
    
        \multirowcell{3}{Portfolio \\ Optimization \\ with VQE}       
        & 4 & 134 & 8 & 32 &  0.09 & 12.58 & \multirowcell{3}{\ding{53}} & \multirowcell{3}{$\star$} & 0.062 & 355.53 \\
        & 8 & 412 & 0 & 64 & 19.83 & 12.58 &  & &  0.17 & 356.44 \\
        & 16 & 1400 & 0 & 128 & \timeout & $\star$  &  & & 84.2 & 1069.14 \\
        \hline
        \multirowcell{3}{QPE\\exact}       
        & 16 & 1119 & 0 & 400 & \multirowcell{3}{\timeout} & \multirowcell{3}{$\star$} & \multirowcell{3}{\ding{53}} & \multirowcell{3}{$\star$} & 5.39 & 638.031 \\
        & 32 & 3775 & 0 & 1312 &   & & & & \timeout & $\star$ \\
        & 64 & 13695 & 0 & 4672 &  & & & & \timeout & $\star$ \\
        \hline
    
        \multirowcell{3}{QPE\\inexact}       
        & 16 & 532 & 0 & 225 & \multirowcell{3}{\timeout} & \multirowcell{3}{$\star$} & \multirowcell{3}{\ding{53}} & \multirowcell{3}{$\star$} & 5.61 & 635.56 \\
        & 32 & 3775 & 0 & 1312  &   & & & &  \timeout & $\star$\\
        & 64 & 13695 & 0 & 4672  &   & & & & \timeout & $\star$\\
        \hline
    
        \multirowcell{3}{Quantum\\ Walk \\ (v-chain)} 
        & 5 &  325 & 72 & 64 & 6.88 & 12.58 & 0.038 & 12.45 & 0.04 & 355.43\\
        & 11 & 1585  & 360 & 344 & \timeout & $\star$& 0.14 & 12.48  & 0.26 & 355.73 \\
        & 17 & 3975  & 24 & 1624 & \timeout & $\star$& 0.30 & 12.45 & 0.26 & 367.69\\   
        \hline
        \multirowcell{3}{Real Amplitudes \\ ansatz with \\Random \\ Parameters}        
        & 16 & 680 & 0 & 128 & \multirowcell{4}{\timeout} & \multirowcell{4}{$\star$}  & \multirowcell{4}{\timeout} & \multirowcell{4}{$\star$} & \multirowcell{4}{\timeout} & \multirowcell{4}{$\star$} \\
        & 32 & 2128 & 0 & 256 &   & & & &  & \\
        & 64 & 7328 & 0 & 512 &   & & & &  & \\
        & 128 & 26944 & 0 & 1024 &   & & & &  & \\
        \hline
    
        \multirowcell{3}{Efficient SU2 \\ ansatz with \\ Random Parameters}        
        & 16 & 680 & 0 & 289 & \multirowcell{3}{\timeout} & \multirowcell{3}{$\star$} & \multirowcell{3}{\timeout} & \multirowcell{3}{$\star$} &53.58 & 1114.42  \\
        & 32 & 2128 & 0 & 576 &   & & & & \timeout & $\star$\\
        & 64 & 7328 & 0 & 1141 &   & & & & \timeout & $\star$\\
        \hline
        \multirowcell{4}{Two Local ansatz \\ with Random \\ Parameters}        
        & 16 & 680 & 0 & 128 & \multirowcell{4}{\timeout} & \multirowcell{4}{$\star$} & \multirowcell{4}{\ding{53}} & \multirowcell{4}{$\star$} & 42.61 & 1120.14 \\
        & 32 & 2128 & 0 & 256 &   & & & & \timeout & $\star$ \\
        & 64 & 7328 & 0 & 512 &   & & & & \timeout & $\star$ \\
        & 128 & 26944 & 0 & 1024 &   & & & & \timeout & $\star$ \\
        \hline
    
       \end{tabular}
    \end{table}
    
    \begin{table}[!h]
    \setlength{\tabcolsep}{1pt} 
    \caption{Results of verifying non-scalable Benchmarks in MQTbench.}
    \label{tab:fullres2}
      \centering
      \scriptsize
      \begin{tabular}{ c | r r r r | c | c | c | c | c | c}
        \hline
        \multirowcell{2}{Algorithm} & \multirowcell{2}{n} & \multirowcell{2}{$|G|$} & \multirowcell{2}{$T$} & \multirowcell{2}{$R$} & \multicolumn{2}{c|}{WMC} &  \multicolumn{2}{c|}{ZX}  &  \multicolumn{2}{c}{CFLOBDD}  \\
        \cline{6-11}
        &  & & & & time(sec) & RSS(MB) &  time(sec) & RSS(MB)  &  time(sec) & RSS(MB)  \\
        \hline
    \multirowcell{3}{Ground\\State}        
        & 4 & 120 & 0 & 48 & 0.075 & 12.58 & \multirowcell{3}{\ding{53}} & \multirowcell{3}{$\star$} & 0.05 & 355.89 \\
        & 12 & 648 & 0 & 144 & \timeout & $\star$ & & & 1.99 & 393.02\\
        & 14 & 840 & 0 & 168 & \timeout & $\star$ & & & 10.60 & 674.30\\
        \hline
        \multirowcell{4}{Pricing\\ Call}       
        & 7 & 339 & 48 & 127 & \multirowcell{4}{\timeout}  & \multirowcell{4}{$\star$} & \multirowcell{4}{\ding{53}} & \multirowcell{4}{$\star$} & 0.06 & 355.95 \\
        & 11 & 727 & 96 & 255 &   & & & &  0.11 & 355.48 \\
        & 15 & 1547 & 144 & 527 &   & & & &  0.32 & 362.28 \\
        & 19 & 4101 & 192 & 1375&   & & & &  2.36 & 386.5\\
        \hline
        \multirowcell{4}{Pricing\\ Put}       
        & 7 & 341 & 48 & 127 & \multirowcell{4}{\timeout} & \multirowcell{4}{$\star$} & \multirowcell{4}{\ding{53}} & \multirowcell{4}{$\star$} &0.06 & 355.95\\
        & 11 & 733 & 96 & 255 &   & & & &   0.11 & 355.55\\
        & 15 & 1565 & 144 & 527 &   & & & &  0.32 & 360.30 \\
        & 19 & 4119 & 192 & 1375 &   & & & & 2.38 & 387.09 \\
        \hline
        \multirowcell{3}{Routing}       
        & 2 & 43 & 0 & 16 & 0.05 & 12.55 & \multirowcell{3}{\ding{53}} & \multirowcell{3}{$\star$} & 0.03 & 346.30\\
        & 6 & 135 & 96 & 48 & 0.34 & 12.39 &  & &  0.05 & 344.56 \\
        & 12 & 135 & 96 & 48 & 60.51 & 12.63 & & &  1.40 & 368.39 \\
        \hline
        \multirowcell{2}{Shor}       
        & 18 & 29935 & 128 & 11281 & \multirowcell{2}{\timeout} & \multirowcell{2}{$\star$} & \multirowcell{2}{\timeout} & \multirowcell{2}{$\star$} & \multirowcell{2}{\timeout} & \multirowcell{2}{$\star$} \\
        & 18 & 29935 & 128 & 11185 &   & & & &  &  \\
        \hline
        \multirowcell{3}{Travelling\\Salesman}       
        & 4 & 165 & 0 & 48 & 0.13 & 12.55 & \multirowcell{3}{\ding{53}} & \multirowcell{3}{$\star$} & 0.04 & 354.64 \\
        & 9 & 390 & 0 & 108 & 2.77 & 12.59 & & & 0.2 & 362.59 \\
        & 16 & 705 & 0 & 192& 20.18 & 12.61 & & & 27.73 & 687.38 \\
        \hline
       \end{tabular}
    \end{table}

\end{document}